%% file: main.tex
\documentclass[conference]{IEEEtran}
\IEEEoverridecommandlockouts

\def\BibTeX{{\rm B\kern-.05em{\sc i\kern-.025em b}\kern-.08em
    T\kern-.1667em\lower.7ex\hbox{E}\kern-.125emX}}

\usepackage{amsmath}
\usepackage{amsthm}
\usepackage{amsfonts}
\usepackage{algorithmic}
\usepackage{algorithm}
\usepackage{array}
\usepackage{subcaption}
\usepackage{textcomp}
\usepackage{stfloats}
\usepackage{float}
\usepackage{url}
\usepackage{verbatim}
\usepackage{graphicx}

\usepackage{tcolorbox}
\usepackage{ulem}
\usepackage{colortbl}
\usepackage{xcolor}
\usepackage{color}
\usepackage{subcaption}
\usepackage{makecell}
\usepackage{booktabs}
 \usepackage{multirow}
 \usepackage{threeparttable}
 \usepackage{balance}
 \usepackage{url}
 \usepackage{bm}
 \usepackage[misc]{ifsym}
 \usepackage{orcidlink}
 \usepackage{hyperref}
\usepackage[justification=centering]{caption}
\renewcommand{\algorithmicrequire}{\textbf{Input:}}

\newtheorem{definition}{Definition}
\newtheorem{theorem}{Theorem}

\begin{document}

\title{Guard-GBDT:  Efficient Privacy-Preserving Approximated GBDT Training on Vertical Dataset
\thanks{
\textsuperscript{\Letter}Corresponding author.} 
}
\author{\IEEEauthorblockN{Anxiao Song}
\IEEEauthorblockA{
\textit{Xidian University}\\
Xi'an, China \\
songanxiao@stu.xidian.edu.cn}
\and
\IEEEauthorblockN{Shujie Cui}
\IEEEauthorblockA{
\textit{Monash University}\\
Melbourne, Australia \\
shujie.cui@monash.edu}
\and
\IEEEauthorblockN{Jianli Bai}
\IEEEauthorblockA{
\textit{Singapore Management University}\\
Singapore \\
baijianli0812@gmail.com}
\and
\IEEEauthorblockN{Ke Cheng\textsuperscript{\Letter}}
\IEEEauthorblockA{
\textit{Xidian University}\\
Xi'an, China \\
chengke@xidian.edu.cn}
\and
\IEEEauthorblockN{\hspace{4cm}}
\IEEEauthorblockA{\hspace{4cm}}
\and
\IEEEauthorblockN{Yulong Shen}
\IEEEauthorblockA{ \textit{Xidian University}\\
Xi'an, China \\
ylshen@mail.xidian.edu.cn}
\and
 \IEEEauthorblockN{Giovanni Russello}
 \IEEEauthorblockA{
\textit{University of Auckland}\\
Auckland, New Zealand \\
g.russello@auckland.ac.nz}
\and
\IEEEauthorblockN{\hspace{2.5cm}}
\IEEEauthorblockA{\hspace{2.5cm}}
}

\maketitle
\begin{abstract}
In light of increasing privacy
concerns and stringent legal regulations, using secure multiparty computation (MPC) to enable collaborative GBDT model
training among multiple data owners has garnered significant
attention. Despite this, existing MPC-based GBDT frameworks face efficiency challenges due to high communication costs and the computation burden of non-linear operations, such as division and sigmoid calculations. In this work, we introduce Guard-GBDT, an innovative framework tailored for efficient and privacy-preserving GBDT training on vertical datasets. Guard-GBDT bypasses MPC-unfriendly division and sigmoid functions by using more streamlined approximations and reduces communication overhead by compressing the messages exchanged during gradient aggregation. We implement a prototype of Guard-GBDT and extensively evaluate its performance and accuracy on various real-world datasets.
The results show that Guard-GBDT outperforms state-of-the-art HEP-XGB (CIKM'21) and SiGBDT (ASIA CCS'24) by up to $2.71\times$ and $12.21 \times$ on LAN network and up to $2.7\times$ and $8.2\times$ on WAN network.
Guard-GBDT also achieves comparable accuracy with SiGBDT and plaintext XGBoost (better than HEP-XGB ), which exhibits a deviation of $\pm1\%$ to $\pm2\%$ only. Our implementation code is provided
at \url{https://github.com/XidianNSS/Guard-GBDT.git}.
\end{abstract}

\input{1-introduction}

\input{2-reletedworks}

\input{3-preliminaries}

\input{4-Problem-Statement}

\input{5-Our_FSS-GBDT}
\input{6-expriments}
\input{7-conclusion}
\bibliographystyle{IEEEtran}
\bibliography{paper}
\newpage
\input{Appendx}

\end{document}

%% file: 1-introduction.tex
\section{Introduction}
\textbf{Gradient Boosting Decision Tree} (GBDT) and its variants such as XGBoost\cite{chen2016xgboost} and LightGBM\cite{ke2017lightgbm} are tree-based machine learning (ML) algorithms known for their high performance and strong interpretability and have been widely used to increase productivity in industries such as finance\cite{wang2022corporate}, recommendation services\cite{behera2022xgboost}, and threat analysis\cite{feng2020xgboost}. 
To build a high-quality GBDT model, there is increasing interest in collaborative training across multiple parties. In practice, the parties in different domains might hold vertically partitioned data, possessing different feature sets for the same group of individuals.  
For example, an insurance company may aim to improve its personalized health insurance recommendations by partnering with a hospital. The insurer maintains rich financial records, including policy coverage, payment history, and claims data, while the hospital holds sensitive medical information such as patient history, treatment details, and current health status. Combining these complementary datasets could significantly enhance the relevance and accuracy of insurance recommendations by aligning financial capacity with medical needs. However, achieving such integration is challenging due to stringent privacy concerns and legal frameworks~\cite{el2023privatree}.



Secure multiparty computation (MPC) is a technique that allows multiple parties to compute a function while maintaining data privacy. In recent years, a series of privacy-preserving works~\cite{lindell2000privacy, liu2020boosting, de2014practical,tian2023sf,abspoel2021secure,xu2024elxgb,wu2020privacy,fang2021large,lu2023squirrel,jiang2024sigbdt} have used MPC to train privacy-preserving GBDT models, and they work for either horizontally~\cite{liu2020boosting,lindell2000privacy,tian2023sf,de2014practical,abspoel2021secure} or vertically~\cite{xu2024elxgb,wu2020privacy,fang2021large,lu2023squirrel,jiang2024sigbdt} distrbuted datasets. {MPC-based horizontal GBDT works for row-wise distributed data where each party holds a subset of samples with identical features (e.g., medical records owned by different hospitals). It typically requires secure feature discretization and secure row-wise feature re-sorting over encrypted data to ensure privacy and correctness. However, this approach cannot be directly extended to vertical GBDT frameworks due to structural differences in data partitioning and feature ownership, and vice versa. 
In vertical GBDT, the data are distributed in column-wise and each party possesses independent features, thus feature re-sorting is unnecessary. Moreover, the split information for a tree node can be revealed to the party owning the corresponding feature. This allows each party to perform certain operations, such as feature discretization (converting continuous values into discrete bins or categories) and comparing features with split thresholds, over unencrypted data. 
In this context, the primary focus of vertical GBDT is on efficiently determining the best feature split at each tree node. 

\begin{table*}[t]
\centering
\footnotesize
\caption{Comparison with existing
 privacy-preserving vertical GBDT works.}
\label{tb:comparison_with_existing_works}
\resizebox{\textwidth}{!}{
\begin{threeparttable}
\begin{tabular}{lcllllllc}
\toprule
\textbf{Approach} & \textbf{System model} & \textbf{Primitive} & \textbf{Communication Round} & \textbf{Communication cost} & \textbf{Computation cost} & \textbf{Leakage} \\
\midrule
ELXGB\cite{xu2024elxgb} & 2PC &PHE & $\mathcal{O}(BF)$ & $\mathcal{O}(BNF\log qp)$&$\mathcal{C}_{M_1}\mathcal{O}(BFN)$  & $G$, $H$, Tree shape \\
\rowcolor{gray!20}
Pivot\cite{wu2020privacy}  & 3PC & PHE, ASS, GC& $\mathcal{O}(BF(\ell^2+\ell+\log\ell))$ & $O(BNF(\log qp + \ell^2))$ &{$(\mathcal{C}_{M_1}+\mathcal{C}_{M_2}+\mathcal{C}_{c_1}+\mathcal{C}_{d})\mathcal{O}(BNF)$}&Tree shape \\
HEP-XGB\cite{fang2021large}  & 2PC & PHE, ASS& $\mathcal{O}(F\ell^2+\log(BF\ell))$ & $\mathcal{O}(BNF(\log q^2p+ \ell^2))$&{$(\mathcal{C}_{M_1}+\mathcal{C}_{M_2}+\mathcal{C}_{c_2}+\mathcal{C}_{d})\mathcal{O}(BNF)$} &  Tree shape \\
\rowcolor{gray!20}
Squirrel\cite{lu2023squirrel}  & 2PC & FHE, ASS & $\mathcal{O}(F\ell^2+\log(BF\ell))$ & $\mathcal{O}(BNF(\log_2q+\ell^2))$&{$(\mathcal{C}_{M_3}+\mathcal{C}_{M_2}+\mathcal{C}_{d}+\mathcal{C}_{c_2})\mathcal{O}(BF)$}
& Tree shape \\
SiGBDT\cite{jiang2024sigbdt} & 2PC &FSS, ASS & $\mathcal{O}(F\ell^2+\log(BF))$ & $O(BNF\ell^2)$ &{$(\mathcal{C}_{s}+\mathcal{C}_{M_2}+\mathcal{C}_{d}+\mathcal{C}_{c_3})\mathcal{O}(BF)$} &Tree shape \\
\rowcolor{gray!20}
\textbf{Guard-GBDT} &\textbf{2PC} &\textbf{FSS, ASS} & $\bm{\mathcal{O}(F+\log(BF))}$ & $\bm{\mathcal{O}(BNF\ell)}$ &{$\bm{(\mathcal{C}_{M_2}+\mathcal{C}_{c_3})\mathcal{O}(BF)}$}&\textbf{Tree shape} \\
 \bottomrule
\end{tabular}
The communication and computation costs summarised in the table are for training one tree node.
$G$ and $H$ represent aggregated statistics corresponding to the first-order and second-order gradients, respectively; $\ell$ is the bit-length of additive secret sharing; $q$ and $p$ are two large prime numbers in homomorphic encryption.  $B$, $F$, and $N$ are the number of buckets, features, and samples, respectively;
$\mathcal{C}_{M_1}$, $\mathcal{C}_{M_2}$ and $\mathcal{C}_{M_3}$ are computation costs of multiplication in PHE, ASS, and FHE, respectively;
$\mathcal{C}_{c_1}$, $\mathcal{C}_{c_2}$ and $\mathcal{C}_{c_3}$ are computation costs of comparison in Garbled Circuits (GC), ASS, and FSS, respectively;
$\mathcal{C}_{s}$ and  $\mathcal{C}_{d}$ are computation costs of the $select$ protocol of FSS and the division protocol, respectively; All other operations are local operations, and the computational cost is negligible.
\end{threeparttable}
}
\end{table*}

\begin{figure}
    \centering
\includegraphics[width=0.9\linewidth]{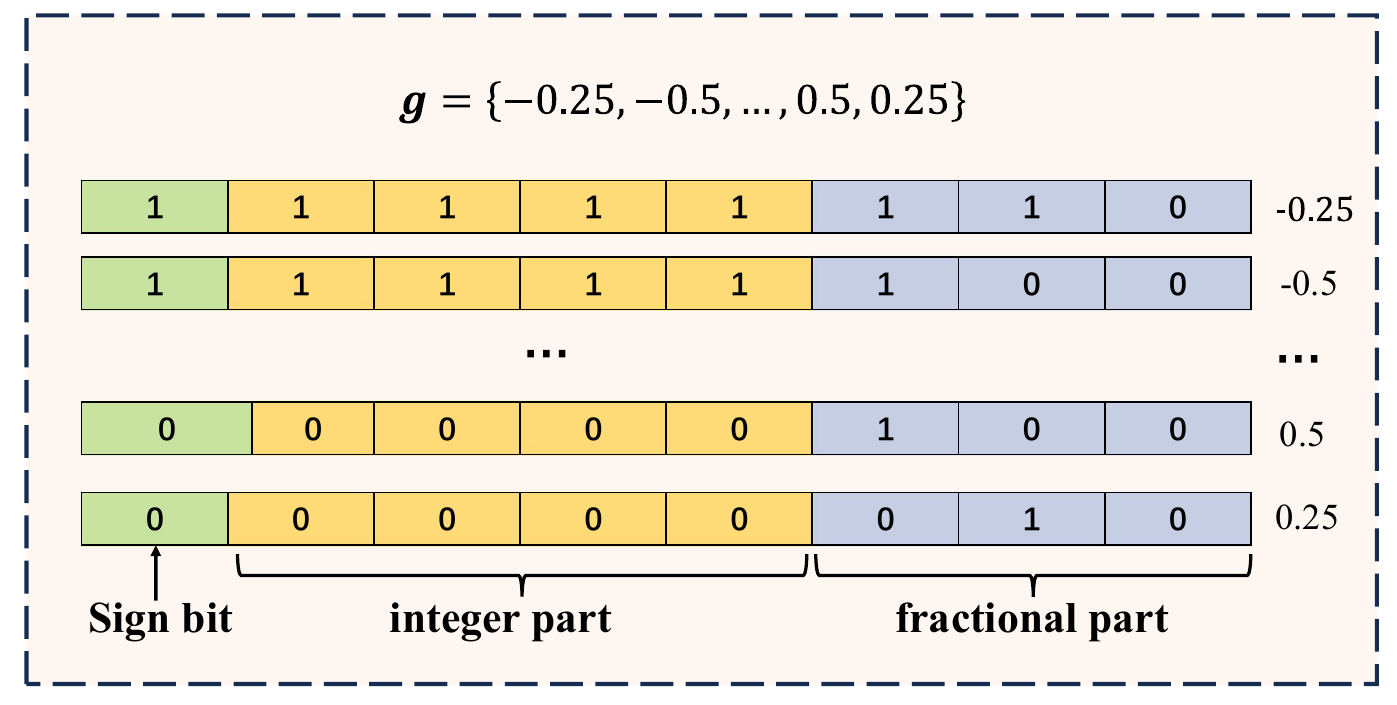}
\captionsetup{justification=justified, singlelinecheck=false}
\caption{Fixed-point encoding for the first-order gradients. In plaintext, $g_i \in \textbf{g} $ is in the range $ [-1,1]$. In MPC,  $g_i$ is encoded into a fixed-point number.  The fraction and sign of $g_i$ carry significant data. The integer part of $g_i$ is filled with $0$'s for positive values or $1$'s for negative values, resulting
in redundant information.}
\label{fig:fixed_point_encoding}
\end{figure}

MPC-based frameworks can effectively protect data privacy, yet they are still impractical due to the cost.  One research line is to make the privacy-preserving frameworks more efficient while maintaining the model's accuracy.  This work aims to propose a more efficient MPC-based vertical GBDT approach by addressing the following two issues. 

\noindent\textbf{Data~inflation in private gradient aggregation}.
GBDT is an ensemble of decision trees, where each new tree uses aggregated first-order gradients $\textbf{g}$ and second-order gradients $\textbf{h}$ to guide tree node splits for optimal performance.
In plaintext, $\textbf{g}$ and $\textbf{h}$ range from \([-1,1]\).
In the context of MPC, these gradients however are encoded into $\ell$-bit fixed-point integers with $1$-bit sign and $\ell_f$-bit fraction, where only the fraction and sign carry significant data, as shown in Figure~\ref{fig:fixed_point_encoding}.
Usually, $\ell=64$ or $\ell=32$ and $\ell_f =16$. This encoding method leads to data inflation after encryption, which increases communication overhead.
For example, approaches like HEP-XGB \cite{fang2021large} leverage partially homomorphic encryption (PHE) and additive secret sharing (ASS) to train the GDBT model, resulting in ciphertexts with $64 \times$ the traffic of plaintexts as  PHE expands gradients to 2048-bit integers.
To improve efficiency, Squirrel \cite{lu2023squirrel} utilizes lattice-based fully homomorphic encryption (FHE) with SIMD to pack multiple ciphertexts into one. However, it introduces
frequent conversions between FHE and ASS, leading to additional communication. More recently, SiGBDT \cite{jiang2024sigbdt} introduces function secret sharing (FSS) to reduce communication overhead, using an $\ell$-bit random value $r$ to mask private $\textbf{g}$ and $\textbf{h}$ over the ring $\mathbb{Z}_{2^\ell}$. Despite this, $\ell-\ell_f$ bits of inflation persist.

\noindent\textbf{Inefficient and complicated non-linear functions}.
GBDT algorithm involves complex non-linear operations, such as division and sigmoid functions, for each node training.
Although these operations are easily computed in plaintext, they require substantial effort to be implemented under MPC.
Previous works often utilize Goldschmidt\cite{ercegovac2000improving} or Newton methods\cite{feng2018privacy} to compute a division approximation through many iterations, and employ approximations based on Fourier series\cite{lu2023squirrel} or Taylor polynomials\cite{zheng2023privet} to implement the sigmoid function.
Both of them introduce heavy computation overheads, making the schemes still not practical for large-scale datasets.

\noindent\textbf{Our contributions.} 
This paper proposes Guard-GBDT, a more efficient privacy-preserving framework that addresses the above two issues for training vertical GBDT models.
In particular, similar to most of the previous works~\cite{xu2024elxgb,fang2021large,lu2023squirrel,jiang2024sigbdt}, we focus on a two-party (2PC) setting---the most common and basic scenario in cross-domain collaborations (e.g., between a bank and an e-commerce platform). 
Note that Guard-GBDT can be easily extended to more parties (see details in Section~\ref{subsec:extension}).}
As done in SiGBDT, Guard-GBDT leverages ASS and FSS\cite{song2024secnet,guo2023gfs} to protect data as they can offer lower communication and simpler computations than PHE, FHE, and Garbled Circuits (GC)~\cite{9155279,jiang2024sigbdt}. Unlike SiGBDT, we use lookup-table-based approximations to eliminate inefficient divisions and sigmoid functions. Moreover, we present a novel division-free split gain metric. To further reduce communication, Guard-GBDT compresses the transmitted data during gradient aggregation by focusing on protecting significant data bits only, rather than entire data bits. As a result, as shown in Table~\ref{tb:comparison_with_existing_works}, Guard-GBDT achieves lower communication and computation costs while maintaining accuracy\footnote{
For fairness, we limit our comparison to secure vertical GBDT methods.}. 
Our contributions and techniques can be summarized as follows:
\begin{itemize}
    \item \textbf{MPC-friendly approximations for division operations  and sigmoid functions.}
    Guard-GBDT eliminates the need for sigmoid function calculations and division operations by implementing a lookup table-based approximation. This approximation is used for both sigmoid and leaf prediction weights, making it particularly friendly to MPC.

    \item \textbf{Division-free split gain metric}. Guard-GBDT introduces a novel division-free split gain.  It provides faster and more MPC-friendly operations than traditional division-based methods and renders an order-of-magnitude improvement in the computation of secure GBDT training.
    \item \textbf{Communication-friendly aggregation protocol.}
    Guard-GBDT designs a communication-efficient aggregation protocol that compresses intermediate communication messages by employing dynamically compact random bits to mask only the significant bits in private data, rather than using fixed-length randomness to mask all bits in existing approaches. This strategy reduces unnecessary data transmission, improving both communication efficiency and overall performance.
    \item \textbf{Extensive evaluations.}
    We implement a prototype of Guard-GBDT and extensively evaluate its performance on various real-world datasets. We compare Guard-GBDT with state-of-the-art HEP-XGB and SiGBDT frameworks in both LAN and WAN networks, and run our experiments with various training parameters. The results show that Guard-GBDT outperforms SiGBDT and HEP-XGB by  up to $2.71\times$ and $12.21 \times$ on LAN and  up to $2.7\times$ and $8.2\times$ on WAN. In most cases, Guard-GBDT also achieves better accuracy than SiGBDT and HEP-XGB. 
\end{itemize}

%% file: 2-reletedworks.tex
\section{Related Works}

MPC is a robust tool with strong security guarantees, which is widely utilized in core business areas to construct privacy-preserving computing tasks. Existing works use ASS or homomorphic encryption (HE) protocols to compute linear functions, and garbled circuit (GC) or FSS is used to compute non-linear functions. These works also adopt an offline-online computation paradigm, shifting most computation and communication overhead to the offline stage to achieve a more efficient online phase. Lindell et al.\cite{lindell2000privacy} proposed the first MPC-based GBDT using Oblivious Transfer (OT) and GC on a horizontally partitioned dataset between two parties. However, this approach suffers from heavy online communication rounds and overhead. Subsequently, Hoogh et al.\cite{de2014practical} and Abspoel et al.\cite{abspoel2021secure} utilized ASS protocols to construct privacy-preserving GBDTs, aiming to enhance performance.

To collaboratively build models among different organizations holding data on the same samples but with different features, the MPC-based vertical GBDT framework has recently become a vibrant area of research. SecureBoost\cite{cheng2021secureboost} was first implemented using Paillier HE and GC to train a GBDT model on two-party vertical datasets, where the gradients are summed using Paillier HE by the passive parties and decrypted by the active party. However, it incurs high communication and computational costs during each training epoch and potentially leaks data labels through the weights sent to the passive party. Zhu et al.~\cite{xu2024elxgb} introduced a novel vertical GBDT framework, ELXGB, where clients locally compute partial gradient sums, encrypt them using HE, and forward them to an aggregation server. To safeguard data labels, noise is added via Differential Privacy (DP). Despite these protections, data labels might still be recoverable through model inversion attacks, and the training scheme incurs increasing communication costs over more training epochs. Following these developments, Pivot\cite{wu2020privacy} and HEP-XGB\cite{fang2021large} employed HE and ASS to enhance computational efficiency. Despite these improvements, the operations involved in frequent encryption and decryption operations of HE remain prohibitively expensive for GBDT training. Building on these advancements, Squirrel\cite{lu2023squirrel} and SiGBDT\cite{jiang2024sigbdt} have utilized lattice-based homomorphic encryption (Learning with Errors and its ring variant) and  FSS to further optimize communication and computational complexity, thus achieving superior HE performance. However, transmitting their encrypted gradients still imposes a significant burden on network bandwidth due to data inflation after encryption.

%% file: 3-preliminaries.tex
\section{Preliminaries}
\subsection{Additive Secret Sharing}
\label{subsec:ass}
In this work, we use 2-out-of-2 additive secret sharing\cite{catrina2010secure} (ASS) to implement secure linear operations. 
ASS includes two different types of secret sharing: one is arithmetic secret sharing $\langle x \rangle$,  which is used for linear arithmetic operations; the other is Boolean secret sharing $[\![x]\!]$, which is applied for Boolean operations.  Each type in ASS are as follows:

\noindent\textbf{Arithmetic sharing.} 
For a given $\ell$-bit value $x\in \mathbb{Z}_{2^{\ell}}$, its arithmetic sharing is denoted as $\langle x\rangle =(\langle x\rangle_0, \langle x\rangle_1)$, where $\langle x\rangle_b$ is held by party $P_b$, $x = (\langle x\rangle_0 + \langle x\rangle_1) \mod 2^{\ell}$, and $b \in\{0, 1\}$. 
For brevity, we omit the operation of \textit{mod}~$2^{\ell}$  in the rest of the paper. 
We have the following secure primitives:
\begin{itemize} 
\item $\langle x\rangle =\textbf{Share}(x)$: The party holding $x$, say $P_b$, chooses a random value $r_x\in \mathbb{Z}_{2^{\ell}}$, set $\langle x\rangle_b = x-r_x \mod 2^{\ell}$, and send $r_x$ to the other party $P_{1-b}$ as $\langle x \rangle_{1-b}$.

    \item $x =\textbf{Open}(\langle x\rangle)$: Given $\langle x \rangle$, each party $P_b$ exchanges their respective secret sharing shares $\langle x\rangle_b$ to recover $x$ by doing $x=\langle x\rangle_0 +\langle x\rangle_1$. 
    
    \item  $\langle z\rangle = \langle x\rangle+\langle y\rangle$: Given $\langle x\rangle$ and $\langle y\rangle$, each party $P_b$ gets the secret sharing of $z =x+y$ by locally computing $\langle z\rangle_b = \langle x\rangle_b+\langle y\rangle_b$. 
    \item  $\langle z\rangle= \langle x\rangle+c$: Given $\langle x\rangle$ and a constant $c$,  each party $P_b$ gets the secret sharing of $z=x+c$ by locally computing $\langle z\rangle_b = \langle x\rangle_b + b\cdot c$. 
    \item $\langle z\rangle =\langle x\rangle\cdot\langle y\rangle$: Given $\langle x\rangle$ and $\langle y\rangle$, the secret sharing of $z=x\cdot y$ can be computed using Beaver multiplication triplet technique~\cite{beaver1995precomputing}, i.e., $\langle r_z\rangle =\langle r_x\rangle\cdot\langle r_y\rangle$ are pre-generated in offline phase by $P_0$ and $P_1$. In online phase, each party $P_b$ locally computes masked $\langle \hat{x} \rangle_b = \langle x\rangle_b-\langle r_x \rangle_b$ and $\langle \hat{y} \rangle_b = \langle y\rangle_b -\langle r_y \rangle_b$. Next, $P_0$ and $P_1$ run $\hat{x} =\textbf{Open}(\langle \hat{x}\rangle)$ and $\hat{y} =\textbf{Open}(\langle \hat{y}\rangle)$.  Finally, each party $P_b$ computes $\langle z \rangle_b =b\cdot \hat{x}\cdot \hat{y}+\hat{y}\cdot\langle x\rangle_b+\hat{x}\cdot \langle y\rangle_b+\langle r_c\rangle_b$.
    It requires $1$ round and $2\ell$ bits of communication to exchange shares of $\hat{x}$ and $\hat{y}$.
    
    \item  $\langle z\rangle = c\cdot\langle x\rangle$: Given $\langle x\rangle$ and a constant $c$, $P_b$ gets the secret sharing of $z=c\cdot x$ by locally computing $\langle z\rangle_b = c\cdot\langle x\rangle_b$.
   
\end{itemize}

\noindent\textbf{Boolean sharing} 
For a one-bit value $x$, its Boolean sharing is denoted as $[\![x]\!] = ([\![x]\!]_0, [\![x]\!]_1)$, where $[\![x]\!]_b$ is held by party $P_b$ and $x = [\![x]\!]_0 \oplus [\![x]\!]_1$.
Boolean sharing is similar to arithmetic secret sharing but with a key difference: in Boolean sharing, the addition ($+$) and multiplication ($\cdot$) are replaced by bit-wise operations $\oplus$ (XOR) and $\land$ (AND).

\subsection{Function Secret Sharing}
\label{sec:fss}
Our work also employs a two-party secret sharing (FSS) scheme to construct secure comparison operations.
FSS can split a private function $f(x)$ into succinct function keys such that every key does not reveal private information about $f(x)$. If each key is evaluated at a specific value \( x \), each party \( P_b \) (\( b \in \{0,1\} \)) can produce a secret share \( \langle f_b(x) \rangle_b \) corresponding to the private function \( f(x) \). The two-party FSS-based scheme is formally described as follows.
\begin{definition}[Syntax of FSS\cite{boyle2016function}]
\label{fSS-Sysntax}
A function secret sharing (FSS) scheme consists of two algorithms, that is, \textbf{FSS.Gen} and \textbf{FSS.Eval}, with the following syntax:
\begin{itemize}
    \item $ (\mathcal{K}_0,\mathcal{K}_1) \leftarrow \textbf{FSS.Gen}(1^{\lambda}, f)$: Given a function description of $f$ and a security parameter $1^{\lambda}$, outputs two FSS secret keys $\mathcal{K}_0, \mathcal{K}_1$. 
    \item $\langle f(x) \rangle_b \leftarrow \textbf{FSS.Eval}(\mathcal{K}_b, x)$: Given an FSS key $\mathcal{K}_b$ and a public input $x$, outputs a secret share $\langle f(x) \rangle_b$ of the function $f(x)$ such that $f(x) = \langle f(x)\rangle_0 + \langle f(x) \rangle_1$.
\end{itemize}
where its correctness and security are proved in Appendix~\ref{app:cs-fss}.
\end{definition}

\subsection{Threat Model}
Similar to previous privacy-preserving GBDT works\cite{abspoel2021secure,cheng2021secureboost,jiang2024sigbdt,fang2021large,gupta2023sigma,jawalkar2024orca}, Guard-GBDT employs a two-party secure computation in the preprocessing model that has gained significant attention in secure machine learning. That is, two parties, $P_0$ and $P_1$, with inputs $x_0$ and $x_1$, aim to compute a public function $y = f(x_0, x_1)$ without revealing any additional information beyond the output $y$.  The randomness independent of the inputs \(x_0\) and \(x_1\) can be generated offline in several ways--through a secure third party (STP), generic 2PC protocols, or specialized 2PC protocols. 
In this work, we employ the first method. Our proposed protocols follow a standard simulation paradigm against a static and semi-honest probabilistic polynomial-time (PPT) adversary that corrupts one of the two parties, as described in Definition \ref{PPT-definition}.
 
\begin{definition}[Semi-honest Security~\cite{canetti2000security,lindell2017simulate}]
\label{PPT-definition} We assume that $\Pi$ is a two-party protocol for computing a function $f: \{0, 1\}^* \times \{0, 1\}^* \rightarrow \{0, 1\}^* \times \{0, 1\}^*$, where $f=(f_0,f_1)$. 
For input pair $(x,y) \in \{0,1\}^n$, the output pair is $(f_0(x,y), f_1(x,y))$.
The view of $P_b$ during an execution of $\Pi$ on $(x, y)$  is denoted by $view^{\Pi}_b(x,y)=(x, r_b, m_1, \ldots, m_t)$, where $r_b$ represents the internal randomness of $P_b$ and $m_i$ represents the $i$-th message passed between the parties. The output of $P_b$  during an execution of $\Pi$ on $(x, y)$ is denoted by $\mathcal{O}^{\Pi}_b$. Let the joint output of two parties be $\mathcal{O}^{\Pi}=(\mathcal{O}^{\Pi}_0, \mathcal{O}^{\Pi}_1)$.
We say that $\Pi$ privately computes $f(x,y)$ if there exist  PPT simulators $\mathcal{S}_0$ and $\mathcal{S}_1$ such that:
\begin{equation}\label{def1-1}
\{\mathcal{S}_0(x, f_0(x, y)), f(x, y)\}\overset{c}{\equiv} \{view^{\Pi}_0(x, y), \mathcal{O}^{\Pi}(x, y)\}
\end{equation}
\begin{equation}\label{def1-2}
\{\mathcal{S}_1(y, f_1(x, y)),f(x, y)\} 
\overset{c}{\equiv} 
\{view^{\Pi}_1(x, y), \mathcal{O}^{\Pi}(x, y)\}
\end{equation}
where $ \overset{c}{\equiv}$ denotes computational indistinguishability.
\end{definition}

%% file: 4-Problem-Statement.tex
\begin{algorithm}[t]
\footnotesize
\caption{Training algorithm of GBDT}
\label{alg:plain-GBDT}
  \begin{algorithmic}
  \REQUIRE Training samples $\textbf{X}=\{\textbf{x}_0,\ldots,\textbf{x}_{N-1}\}$; Label set of samples $\textbf{y}=\{y_0,\ldots,y_{N-1}\}$; Maximum  tree depth $D$ 
  \ENSURE A GBDT model $\mathcal{M}^{(T)}=\{\mathcal{T}_0,\ldots,\mathcal{T}_{T-1}\}$
\end{algorithmic}

\newcommand{\mybluebox}[1]{\colorbox{blue!10}{#1}}
\newcommand{\myyellowbox}[1]{\colorbox{yellow!30}{#1}}
\begin{algorithmic}[1]

\renewcommand{\algorithmicrequire}{\textbf{Function:}}
\REQUIRE GBDT.TraingModel(\textbf{X}, \textbf{y}, $D$):
\STATE $ \textbf{Bucket} \gets\text{GBDT.Distinct}(\textbf{X})$, $\textbf{y}^{(0)}=\textbf{0}$
\FOR{$t\in [1,T-1]$}
\IF{t=1}
\STATE $\hat{\textbf{y}}^{(0)}=\textbf{0}$
\ELSE
\STATE {$\hat{\textbf{y}}^{(t-1)}= \sum^{T-1}_t \mathcal{T}_t(\textbf{X})$\hspace{3.7em}}
\ENDIF
\STATE \mybluebox{$\textbf{p}=\text{sigmoid}(\hat{\textbf{y}}^{(t-1)})$\hspace{10.10em}}
\STATE \mybluebox{$\textbf{g}=\textbf{y}-\textbf{p}$; $\textbf{h}=\textbf{p}\cdot (1-\textbf{p})$\hspace{7.55em}}
\STATE $\text{GBDT.BulidTree}(\mathcal{T}_t.root,\textbf{X},\textbf{g},\textbf{h},\textbf{Bucket}, D, 0)$
\STATE $\mathcal{M}.\text{push}(\mathcal{T}_t)$
\ENDFOR
\RETURN $\mathcal{M}$
\end{algorithmic}

\begin{algorithmic}[1]
\renewcommand{\algorithmicrequire}{\textbf{Function:}}
\REQUIRE $\text{GBDT.BulidTree}(\mathcal{T}.root ,\textbf{X},\textbf{g},\textbf{h}, \textbf{Bucket}, D, d)$:
\setcounter{ALC@line}{13}
\IF{ current deepth $d < D$} 
\STATE  $(z_*,u_*) \gets \text{GBDT.BestSplit}(\textbf{X})$; $\mathcal{T}.root.value \gets (z_*,u_*)$
\STATE  $\textbf{X}_L \gets \{\textbf{X}[i] ~|~\textbf{X}[i,z_*]<\textbf{Bucket}[z_*,u_*]\}$;  $\textbf{X}_R\gets \textbf{X} -\textbf{X}_L$
\STATE $\text{GBDT.BulidTree}(\mathcal{T}.root.left\_child,\textbf{X}_L,\textbf{g},\textbf{h}, \textbf{Bucket}, D,d+1)$
\STATE $\text{GBDT.BulidTree}(\mathcal{T}.root.right\_child,\textbf{X}_R,\textbf{g},\textbf{h}, \textbf{Bucket}, D,d+1)$
\ELSE
\STATE\myyellowbox{$G_{X}=\sum_{i\in \textbf{X}} \textbf{g}[i]$; $H_{X} =\sum_{i \in \textbf{X}}\textbf{h}[i]$ \hspace{3.35em}}
\STATE \mybluebox{$\mathcal{T}.leaf.value \gets w = -{G_{X}}/{H_X}$\hspace{4.65em}}
\ENDIF
\end{algorithmic}

\begin{algorithmic}[1]

\renewcommand{\algorithmicrequire}{\textbf{Function:}}
\REQUIRE $\text{GBDT.BestSplit}(\textbf{X}):$
\setcounter{ALC@line}{24}
\FOR{$z \in [1,F]$}
\FOR{$u \in [1,B-1]$}
\STATE  $\textbf{X}_L \gets \{\textbf{X}[i] ~|~\textbf{X}[i,z]<\textbf{Bucket}[z,u]\}$; 
\STATE $\textbf{X}_R\gets \textbf{X}-\textbf{X}_L$;
\STATE  \myyellowbox{$G_{X}=\sum_{i\in \textbf{X}} \textbf{g}[i]$; $H_{X} =\sum_{i\in \textbf{X}}\textbf{h}[i]$\hspace{3.65em}}
\STATE \myyellowbox{$G_{L}=\sum_{i\in \textbf{X}_L}\textbf{g}[i]$; $H_{L}=\sum_{i\in \textbf{X}_L}\textbf{h}[i]$\hspace{2.70em}}
\STATE  \myyellowbox{$G_{R} = G_X -G_L$; $H_R=H_X-H_R$\hspace{3.40em}}
\STATE  \mybluebox{$\mathcal{G}^{(u,z)} = \frac{1}{2}(\frac{({G_L})^2} {H_{L}+\gamma}+\frac{(G_{R})^2}{H_{R}+\gamma}-\frac{(G_{X})^2}{H_{X}+\gamma})$\hspace{1.9em}}
\ENDFOR
\ENDFOR
\STATE $ (z_*,u_*) \gets Argmax(\{ \mathcal{G}^{(1,1)},\ldots,\mathcal{G}^{(F,B-1)}\})$
\RETURN ($z_*,u_*$)
\end{algorithmic}
\end{algorithm}

\subsection{Gradient Boosting Decision Tree}
GBDT is a boosting-based machine learning algorithm that ensembles a sequence of decision trees in an additive manner~\cite{fu2019experimental}. In the model, each decision tree consists of internal nodes, edges, and leaf nodes: each internal node represents a test (e.g., \(Age<20\)) between a feature and a threshold, each edge represents the outcome of the test, and each leaf node represents a prediction weight. In this paper, we assume that all trees are complete binary trees for maximum computational overhead. This means that each tree has $2^D-1$ internal nodes and $2^D$ leaf nodes when the tree depth is denoted as  $D$. Given a sample set $ \textbf{X} \in \mathbb{R}^{N\times F}$ with $N$ samples and $F$ features, each sample $\textbf{X}[i]$ would be classified into one leaf node of each tree, then the GBDT model sums the prediction weights of all trees as the final prediction $\hat{\textbf{y}}[i]=\sum^T_t \mathcal{T}_t(\textbf{X}[i])$, 
where $T$ is the number of decision trees and $\mathcal{T}_t(\textbf{X}[i])$ is the prediction weight, denoted as $w_t$, of the $t$-th tree.



The GBDT training algorithm's main task is to determine the \textit{best-split candidate} (i.e., a pair of the feature and threshold) for each tree node by evaluating all possible thresholds for each feature. However, testing all possible thresholds is computationally intractable in practice. Existing GBDT approaches\cite{chen2016xgboost,wang2022foster} accelerate the training process by discretizing numeric features. For example, the feature $Age \in [0,100]$ can be discretized into three buckets, such as $[0,20]$, $[20,60]$, and $[60,100]$. During training, the algorithm evaluates each bucket boundary (e.g., 20 and 60) sequentially as a potential split candidate. To simplify the presentation, we assume that each feature is discretized into $B$ buckets, resulting in \(F\cdot (B-1)\) split candidates for all features. We use $\textbf{Bucket} \in \mathbb{R}^{F\times(B-1)}$ to represents all potential split candidates, where $\textbf{Bucket}[z,u]$ be the $u$-th threshold of the $z$-th feature.

Algorithm \ref{alg:plain-GBDT} describes the GBDT training process. 
The trees are trained in sequence, where the $t$-th tree is trained based on the predicates of the previous $(t-1)$ trees, and each tree is trained with the full dataset $\textbf{X}$. 
Assume now we are going to train the $t$-th tree. 
Given $\textbf{X}$, we first input them into the first $t-1$ trees to get the predicates for the N data samples, denoted as $\hat{\textbf{y}}^{(t-1)}$. 
Second, we calculate the first-order gradient $\textbf{g}$ and second-order gradient $\textbf{h}$ of all samples as follows:
\begin{equation}
\begin{aligned}
       &\textbf{g}=\partial_{\hat{\textbf{y}}^{(t-1)}}\mathcal{L}(\textbf{y},{\hat{\textbf{y}}}^{(t-1)}); \textbf{h}=\partial^2_{\hat{\textbf{y}}^{(t-1)}}\mathcal{L}(\textbf{y},{\hat{\textbf{y}}}^{(t-1)})
\end{aligned}
\end{equation}
where $\mathcal{L}(\cdot)$ is a task loss function that measures the difference between the predicted label $\hat{\textbf{y}}^{(t-1)}$ and the ground label $\textbf{y}$. 
Without loss of generality, we focus on the binary classification tasks that use cross-entropy loss function, which are more complex than regression tasks\footnote{The existing private GBDT for regression tasks usually uses MSE loss such that  $\textbf{g}=\hat{\textbf{y}^{(t-1)}}-\textbf{y}$, \textbf{h}=\textbf{1}, avoiding the necessary complex division.}.  In this task, the GBDT model’s prediction is represented as:
\begin{equation}
\label{eq:sigmoid}
    \textbf{p}=\text{sigmoid}(\hat{\textbf{y}}^{(t-1)})=\frac{1}{1+\exp(-\hat{\textbf{y}}^{(t-1)})}  
\end{equation}
Thus, we have the first-order gradient $\textbf{g}=\textbf{y}-\textbf{p}$ and the second-order gradient $\textbf{h}=\textbf{p}\cdot (1-\textbf{p})$. 

The $t$-th tree is trained top-down in a recursive fashion. The main task of training is to assign a feature-threshold pair to each internal node based on a gradients-guided criterion, which determines how well the pair classifies the current samples, and assign a value to each leaf. 
Starting with the root node, to evaluate a split candidate $\textbf{Bucket}[z,u]$, we first partition the data samples into subset $\textbf{X}_L = \{\textbf{X}[i] \mid \textbf{X}[i,z] \leq $\textbf{Bucket}[z,u]$\}$ and $\textbf{X}_R=\textbf{X}-\textbf{X}_L$, where $\textbf{X}[i,z]$ represents the $z$-th feature of the $i$-th sample. Next, we aggregate the gradient statistics of parent and child nodes as follows:
\begin{equation}
\label{eq:aggregate}
\begin{aligned}
  &G_{X}=\sum_{i\in \textbf{X}}{\textbf{g}[i]};\quad\quad H_{X}=\sum_{i\in \textbf{X}}\textbf{h}[i]\\
  & G_{L}=\sum_{i\in \textbf{X}_L}{\textbf{g}[i]};\quad\quad H_{L}=\sum_{i\in \textbf{X}_L}\textbf{h}[i]\\
  & G_{R}=\sum_{i\in \textbf{X}_R}{\textbf{g}[i]};\quad\quad H_{X}=\sum_{i\in \textbf{X}_R}\textbf{h}[i]
\end{aligned}
\end{equation} 



Finally, we can evaluate a gain $\mathcal{G}$ for the  split candidate $\textbf{Bucket}[z,u]$ as follows:
\begin{equation}
\label{eq:gain}
    \mathcal{G} = \frac{1}{2}(\frac{({G_L})^2}{H_{L}+\gamma}+\frac{(G_{R})^2}{H_{R}+\gamma}-\frac{(G_{X})^2}{H_{X}+\gamma})
\end{equation}
where $\gamma>0$ is a (public) constant denoted as regularization parameters.
The GBDT algorithm iterates all $(z,u)$ pairs to find the split candidate that gives the maximum gain. 
Once the best feature and threshold are determined, the samples will be split into two partitions accordingly and used to train the nodes in the next level, and so forth. 
When a GBDT tree $\mathcal{T}_t$ stops growing, the prediction weight $w$ for a leaf node can be computed as:
\begin{equation}
\label{eq:leaf}
w = - \frac{G_{Leaf}}{H_{Leaf}+\gamma}
\end{equation}
where $G_{Leaf}$ and $H_{Leaf}$ denotes gradient statistics of the leaf.

%% file: 5-Our_FSS-GBDT.tex
\section{Guard-GBDT Construction} 
\label{sec:split}
This work focuses on training a GBDT model securely between two untrusted parties who share the dataset vertically. Specifically, a training dataset $\textbf{X}=\textbf{X}_0\|\textbf{X}_1\in \mathbb{R}^{N\times F}$ is vertically partitioned, and the party $P_b$ (\(b \in \{0,1\}\)) locally holds a dataset $\textbf{X}_b\in \mathbb{R}^{N\times F_b}$ with $N$ samples and $F_b$ features. For simplicity, it assume that the features held by $P_0$ come before those held by $P_1$, and $P_1$ holds the whole labels $\textbf{y}$.

For the parties who do not trust each other, all steps of GBDT training should be performed secretly without leaking anything to the other party. Basically, after training, each party $P_b$ can only know the tree parts that involve the features owned by $P_b$. Any data owned or generated by one party should be hidden from the other. For instance, how the data samples are partitioned by the node owned by $P_b$ should be unknown to $P_{1-b}$. 
MPC can be employed to achieve security goals, where the data that needs to be shared between the two parties is protected with secret sharing techniques. 
It causes data inflation and inefficiency in processing non-linear functions. Guard-GBDT is an MPC-friendly and communication-efficient framework that solves these issues.

\subsection{Overview of Guard-GBDT} 
To streamline the use of non-linear functions,  
inspired by gradient quantization techniques, we introduce a segmented division-free approximation for the sigmoid function and leaf weights. Each segment’s approximation is stored in a lookup table, requiring $(n+1)\ell$ bits of memory, where $n$ is a constant representing the number of segments. These tables can be processed offline, allowing online execution of the approximation with minimal communication—-only one communication round and $n\ell$ bits of overhead are needed. 
Note that the overhead of previous works\cite{jiang2024sigbdt,lu2023squirrel,fang2021large} for the sigmoid function is $\mathcal{O}(\ell^2+\ell)$. Besides, we also propose a division-free split gain to determine the best-split candidate for each tree node. Compared with Eq. \ref{eq:gain}, our gain requires only $5$ communication round and $9\ell$ bits of overhead. To address data inflation, Guard-GBDT introduces compact bit-width random values to protect only the significant data bits, eliminating unnecessary data transmission. This design strategy achieves efficient message compression, minimizing bandwidth usage and saving $\ell-\ell_f-2$ bits of communication cost for each gradient element during aggregation.


\subsection{Division-free Sigmoid Function}
\begin{figure}
    \centering
    \includegraphics[width=0.9\linewidth]{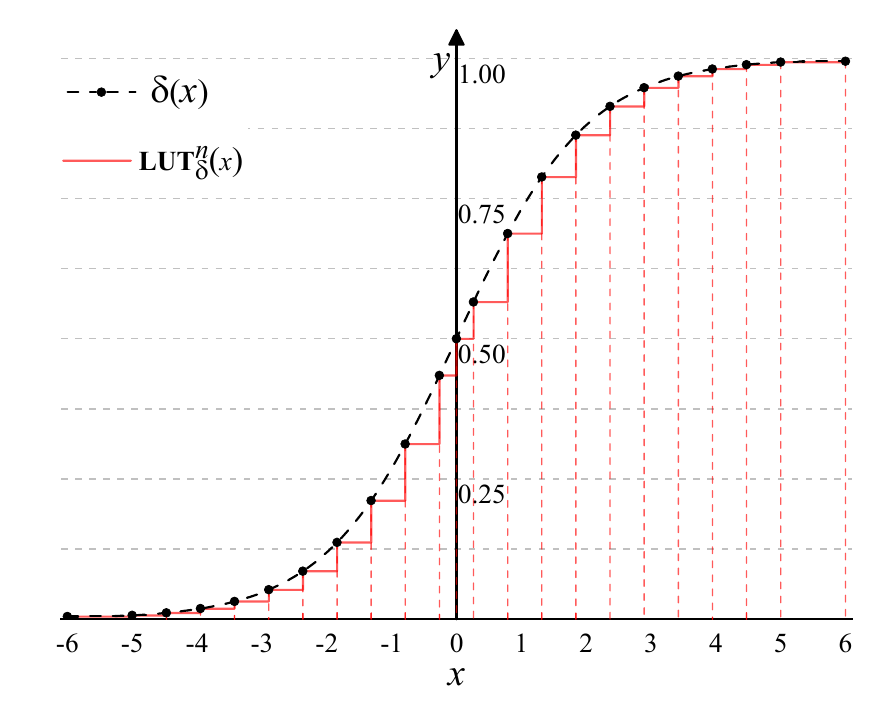}
    \caption{Our \(sigmoid\) function with multiple segments}
    \label{fig:sigmoid-function}
\end{figure}

Guard-GBDT achieves better efficiency with minor compromises in accuracy. To make the sigmoid function division-free, our main idea is to segment the input space into several intervals, approximate the sigmoid value of inputs in the same intervals to the same value, i.e., the minimum, and store them in a table. 
During the training, the parties can get the sigmoid value directly from the table by checking which interval the input belongs to. This process just involves secure comparisons, which is significantly more efficient than other alternatives of division and exponentiation. 

As shown in Fig.~\ref{fig:sigmoid-function}, the $sigmoid~\delta(x)$ function is a smooth, continuous, and monotonic function with outputs strictly bounded within the range ($0,1$). Particularly, $\delta(x)$ saturates for the inputs beyond the interval $(-5, 5)$, approaching $0$ or $1$. As a result, we can approximate it as follows: 
\begin{equation}
\label{eq:secure-sigmoid}
\delta(x) \approx {\delta}^{'}(x) =
    \begin{cases}
     \delta(\omega_0) , & \text{if } x < \omega_0\\
       \delta(\omega_0), & \text{if } x \in [\omega_0, \omega_1)\\
        \cdots \\
        \delta(\omega_{n-1}), & \text{if } x \in [\omega_{n-1},\omega_{n})\\
        \delta(\omega_{n}), & \text{if } x \geq \omega_{n}.
    \end{cases}
\end{equation}
where  $\omega_i = -5+10i/n$, $\delta(x)=sigmoid(x)=1/(1+e^{-x})$, and let $\textbf{LUT}^{n}_{\delta}[i]=\delta(w_i)$ for $i\in \{0,1,\cdots, n\}$.

Fig.~\ref{fig:sigmoid-function} shows the absolute difference between the ground truth $\delta(x)=sigmoid(x)$ and our approximation $\delta'(x)$ using the lookup table. 
Although there is some deviation between our approximation and the original output of the sigmoid function, which might affect the accuracy of a single tree, it is tolerable for training the GBDT model. The quantized training for GBDT\cite{shi2022quantized} has demonstrated that the model's performance can be enhanced by ensembling multiple GBDT trees with low accuracy. Our experiments in Section~\ref{subsec:acc-exp} prove that. 

The $\textbf{LUT}^{n}_{\delta}$ is public for two parties $P_0$ and $P_1$ such that they {only require $(n+1)\ell$ bits to} construct the public look-up table $\textbf{LUT}^{n}_{\delta}$ locally in offline phase. In the online phase, they can merely perform secure table queries and output the corresponding results.
\begin{algorithm}[t]
\footnotesize
\caption{Secure approximated \(sigmoid\) protocol $\Pi_{\text{LUT}_{\delta}^{n}}$}
\label{alg:secure-sigmoid-protocol}
    \begin{algorithmic}
    \REQUIRE The secret share $\langle x\rangle$.
    \ENSURE The secret share $\langle {\delta^{'}}(x) \rangle$
    \end{algorithmic}
\flushleft\underline{$\Pi_{\textbf{LUT}_{\delta}^{n}}$.offline($n$,$\ell$)}
    \begin{algorithmic}[1]
     \FOR{$i \in \{0,1,\cdots n\}$ \textbf{in parallel}}
      \STATE Let $\text{LUT}^{n}_{\delta}[i]=\delta(\omega_i)$ and  $\omega_i = -5+10i/n$.
     \STATE $\alpha_i \overset{\$}{\gets} \mathbb{Z}_{\ell}$
     \STATE  $\mathcal{K}^{i}_b \gets {\text{DCF}}.\text{Gen}(\alpha_i)$ 
     \STATE  $\langle \alpha_i \rangle \gets \text{Share}(\alpha_i)$
     \STATE Send ($\langle \alpha \rangle_b, \mathcal{K}^{i}_b$) into the party $P_b$.
     \ENDFOR
    \end{algorithmic}
    \underline{$\Pi_{\textbf{LUT}_{\delta}^{n}}$.online($\langle x \rangle$)}
    \begin{algorithmic}[1] 
    \setcounter{ALC@line}{7}
     \FOR{$i \in \{1,\cdots n\}$ \textbf{in parallel}}
     \STATE $\delta'(\omega_i) \gets \text{LUT}_{\delta}^{n}[i]$
     \STATE $m_i = x+\alpha_i-\omega_i \gets \text{Open}(\langle x\rangle+\langle \alpha_i\rangle-\omega_i)$
     \label{alg:sec_sigmoid_line_3}
    \STATE $ \langle \beta_i\rangle={\text{DCF}}.\text{Eval}(\mathcal{K}^{i}_b, m_i)$ // check if $x<\omega_i$
    \ENDFOR
    \STATE $\langle \delta^{'}(x)\rangle = \delta(\omega_0)+\sum_{i=1}^{n}\langle \beta_i \rangle (\delta'(\omega_i)-\delta'(\omega_{i-1}))$
    \RETURN $\langle \delta'(x)\rangle$
   \end{algorithmic}
\end{algorithm}
Algorithm~\ref{alg:secure-sigmoid-protocol} shows the detail of our secure approximated \(sigmoid\) protocol denoted as $\Pi_{\textbf{LUT}^{n}_{\delta}}$. The core of this protocol is to privately select an activated segment of Eq.~\ref{eq:secure-sigmoid}. Thus, we use the DCF-based less-than protocol in Section \ref{sec:fss} to check if $\beta_i=x<\omega_i$ in every segment (ref. Line 8-11 in Algorithm~\ref{alg:secure-sigmoid-protocol}). Specifically, we assume that the $i$-th segment is activated and the $j$-th segment is not activated. After checking all segment, we have  $\beta_j = 1$ for all $j < i$ and $\beta_j = 0$ for all $j > i$. To obtain $\delta(\omega_i)$ of the activated segment, we only need to compute $\delta'(x)= \delta(\omega_0)+\sum_{i=1}^{n}\beta_i(\text{LUT}_{\delta}^{n}[i]-\text{LUT}_{\delta}^{n}[i]) = \delta(\omega_0)+\sum_{i=1}^{n}\beta_i(\delta(\omega_i)-\delta(\omega_{i-1}))$ in MPC. During the online phase, only $Open$ operation (ref. Line 10 in Algorithm~\ref{alg:secure-sigmoid-protocol}) necessitates a single communication involving the transmission of $\ell$ bits, whereas other computations are performed locally.




\subsection{Division-free Leaf Weight}
The leaf weight $w_t=-G_{X}/H_{X}$ is fundamental to GBDT training but requires division. To avoid division, we also use a lookup table to approximate the leaf weight, which is similar to the approach applied to our sigmoid function. From the GBDT algorithm, we observe that the prediction ${p}={sigmoid}(\sum_t^{T} w_t)$ offers a practical property: the prediction $p\approx 1$  when $\sum_t^{T} w_t>=5$ and $p\approx0$ when $\sum_t^{T} w_t=<-5$. The property allows each leaf weight $w_t$ to be scaled within the range of approximately ~$[-5,5]$. Thus, we employ a numerical approximation using multiple segments to approximate the leaf weight $w_t$, complemented by a lookup table $\textbf{LUT}_{w}$ to store these approximations. Our approximated leaf weight is given as follows:

    

\begin{equation}
\label{eq:secure-leaf}
w \approx w^{'} =
    \begin{cases}
     -5 , & \text{if } -\frac{G_{X}}{H_{X}} < \omega'_0\\
       \omega'_0, & \text{if } -\frac{G_{X}}{H_{X}} \in [\omega'_0, \omega'_1)\\
        \cdots \\
      \omega'_{n-1} , & \text{if }  -\frac{G_{X}}{H_{X}} \in [\omega'_{n-1},\omega'_{n})\\
       \omega'_n , & \text{if } -\frac{G_{X}}{H_{X}}\geq \omega'_n 
    \end{cases}
\end{equation}
where  $\omega'_i = -5+10i/n$ for $i\in \{0,1,\cdots, n\}$ and let $\textbf{LUT}^{n}_{\delta}[i]=\omega'_{i}$.
\begin{algorithm}
\caption{Secure approximated  leaf weight  protocol~$\Pi_{\text{LUT}_{w}^{n}}$}
\label{alg:secure-leaf-weight}
\footnotesize
    \begin{algorithmic}
    \REQUIRE The secret shares $\langle G_{X}\rangle$ and  $\langle H_{X} \rangle$.
    \ENSURE The secret share  $\langle w' \rangle$
    \end{algorithmic}
    \flushleft  \underline{$\Pi_{\textbf{LUT}_{w}^{n}}$.offline($n$,$\ell$)}
    \begin{algorithmic}[1]
     \FOR{$i \in \{0,1,\cdots n\}$ \textbf{in parallel}}
      \STATE Let $\text{LUT}^{n}_{w}[i]=\omega'_i = -5+10i/n$.
     \STATE $\alpha_i \overset{\$}{\gets} \mathbb{Z}_{2^\ell}$
     \STATE  $\mathcal{K}^{i}_b \gets {\text{DCF}}.\text{Gen}(\alpha_i)$ 
     \STATE  $\langle \alpha_i \rangle \gets \text{Share}(\alpha_i)$
    \STATE Send ($\langle \alpha \rangle_b, \mathcal{K}^{i}_b$) into the party $P_b$.
     \ENDFOR
    \end{algorithmic}
   \underline{$\Pi_{\textbf{LUT}_{w}^{n}}$.online($\langle G_{X}\rangle$, $\langle H_{X} \rangle$)}
    \begin{algorithmic}[1] 
    \setcounter{ALC@line}{7}
     \FOR{$i \in \{1,\cdots n\}$ \textbf{in parallel}}
     
     \STATE $\omega'_i \gets \text{LUT}_{w}^{n}[i]$ 
     \STATE $\langle x\rangle=-\langle G_{I}\rangle-\omega'_i\langle H_{X}\rangle$
     \STATE $x+\alpha_i\gets \text{Open}(\langle x\rangle+\langle \alpha_i\rangle)$
    \STATE $ \langle \beta_i\rangle={\text{DCF}}.\text{Eval}(x+\alpha_i)$ // check if $ -\frac{G_{X}}{H_{X}}<\omega'_i$.
    \ENDFOR
    \STATE $\langle w' \rangle =\omega_0+\sum_{i=1}^{n}\langle \beta_i \rangle(\omega_i-\omega_{i-1})$
    \RETURN  $\langle w' \rangle$
   \end{algorithmic}
\end{algorithm}

Algorithm \ref{alg:secure-leaf-weight} depicts the detail of the $\Pi_{\text{LUT}^{n}_{w}}$ protocol for leaf weight. Its core consists of selecting an activated segment of  Eq.~(\ref{eq:secure-leaf}) in a private manner. To avoid a complex division of $w=-G_{X}/H_{X}<\omega'_i$, we reformulate this computation as $-G_{X}-H_{X}\omega_i<0$. Similar to our sigmoid protocol, $n$ DCF-based comparison protocols are used to check whether $\beta_i=-G_{X}-H_{X}\omega_i<0$ for every segment. After that, we only compute $w'=\omega'_0+\sum_{i=1}^{n}\beta_i(\omega'_i-\omega'_{i-1}) = \text{LUT}_{w}^{n}[0]+\sum_{i=1}^{n}\beta_i( \text{LUT}_{w}^{n}[i]- \text{LUT}_{w}^{n}[i-1])$ in the MPC setting. 

\subsection{Division-free Split Gain}
The essence of the split gain is to select the best-split candidate with the maximum $\mathcal{G}$ among multiple different split candidates.
Without loss of generality, we assume that the data samples \textbf{X} are partitioned into \(\{L_1, R_1\}\) and \{\(L_2, R_2\)\} by two different split candidates. A naive division-free solution\footnote{
Following the equation \ref{eq:gain} and 
the principles of inequality transformation, we have:
\begin{equation}
\label{eq:first-way-gain}
\begin{aligned}
\mathcal{G}_1 >\mathcal{G}_2&\Leftrightarrow\frac{(G_{L_1})^2}{H_{L_1}}+\frac{(G_{R_1})^2}{H_{R_1}}>\frac{(G_{L_2})^2}{H_{L_2}}+\frac{(G_{R_2})^2}{H_{R_2}}\\
&\Leftrightarrow H_{L_2}H_{R_2}(H_{R_1}(G_{L_1})^2+H_{L_1}(G_{R_1})^2)\\
&>H_{L_1}H_{R_1}(H_{R_2}G_{L_2})^2+H_{L_2}(G_{R_2})^2). 
\end{aligned}
\end{equation}
 }~is to convert the comparison between \(\mathcal{G}_1\) and \(\mathcal{G}_2\) into a check that:
 \begin{equation}
 \label{eq:comparison}
 \begin{aligned}
      \mathcal{G}_1 >\mathcal{G}_2&\Leftrightarrow H_{L_2}H_{R_2}H_{R_1}G_{L_1}^2+H_{L_2}H_{R_2}H_{L_1}G_{R_1}^2\\ &>H_{L_1}H_{R_1}H_{R_2}G_{L_2}^2+H_{L_1}H_{R_1}H_{L_2}G_{R_2}^2.
 \end{aligned}
 \end{equation}
 However, each element of the solution requires 8 secure multiplication operations and the exchange of $17\ell$ bits of communication messages. 
To reduce overheads in both computation and communication,  we   propose a new division-free gain, which is given as follows:
\begin{equation}
\mathcal{G}'=
\label{eq:our-gain}
   \begin{cases}
        (H_R+\gamma)(G_{L})^2+(H_L+\gamma)(G_{R})^2 & 2H_L<H_X\\
        -((H_R+\gamma)(G_{L})^2+(H_L+\gamma)(G_{R})^2) &2H_L\geq H_X
   \end{cases}
\end{equation}
where $\{L, R\}$ is one split schema generated by a possible split candidate on an available sample set. Compared to the original $\mathcal{G}$ of Eq. \ref{eq:gain}, our key insight involves replacing the non-linear function $1/H_{@}$ with a simpler linear function and removing the term ${G_X^2}/{(H_X+\gamma)}$, where $@\in \{L, R\}$. This change is based on our observation that $\mathcal{G}'$ is consistent with {$\mathcal{G}$} in Eq.~\ref{eq:gain}, and the term ${G_X^2}/{(H_X+\gamma)}$ is found to be ineffective in comparing between split candidates in the same sample set \textbf{X}, indicating these operations does not affect the choice of the best split. Compared to the conversion of Eq.~\ref{eq:comparison}, our split metric does not depend on multiple divisors from other gains, making it particularly well suited to train a private GBDT model on a large-scale dataset. This independence from other gains facilitates more efficient parallel processing, enhancing both the scalability and efficiency of the training process. {We also theoretically prove that our $\mathcal{G}'$ is equivalent to the original $\mathcal{G}$, where its detailed correctness analysis is  provided in Appendix~\ref{app:correctness-g}}.

\subsection{Communication-efficient Aggregation Protocol}
\label{subsec:agg}
{The gain for each pair is computed using the aggregated gradients of the data samples assigned to the node. In MPC-based solutions, to prevent data exposure, each party employs a secret indicator vector \(\textbf{s} \in \{0,1\}^N\) to denote sample assignments, where \(\textbf{s}[i] = 1\) if the \(i\)-th sample belongs to the current tree node and \(\textbf{s}[i] = 0\) otherwise. Thus, the aggregation of the gradient is the computation of $G=\textbf{s}\cdot\textbf{g}$, followed by a local addition of the vector entries of the results. However, this process incurs data inflation after encryption and increases communication overhead.}
To minimize communication overhead, we propose a secure aggregation protocol $\Pi_{\text{Agg}}$, where the protocol inputs Boolean secret shares $\langle \textbf{s}\rangle=\{\langle s_0\rangle,\ldots,\langle s_{N-1}\rangle\}$, $\ell$-bit secret shares $\langle \textbf{g}\rangle = \{\langle g_0 \rangle,\ldots,\langle g_{N-1}\rangle\}$ and outputs a $\ell$-bit secret shares $\langle G \rangle =\sum_{i=0}^{N-1}\langle s_i\cdot g_i\rangle$. Algorithm \ref{alg:agg} depicts the details of  $\Pi_{\text{Agg}}$ protocol. 

Recall that in a 2PC multiplication protocol of $\langle s_i\cdot g_i\rangle$, three secret shared values $\langle r\rangle, \langle r_{g_i}\rangle,\langle r_{s_i}\rangle$ are pre-generated to achieve an efficient online computation in the offline phase, where $r=r_{s_i}\cdot r_{g_i}$. During the online phase, each party $P_b$ ($b\in \{0,1\}$) computes $\langle \tilde{s}_i \rangle_b = \langle s_i \rangle_b+\langle r_{s_i}\rangle_b$ and $\langle \tilde{g_i} \rangle_b = \langle g_i \rangle_b+\langle r_{g_i}\rangle_b$.
In this way, $s_i$ and $g_i$ are securely masked by random values, allowing the masked  $\tilde{s_i}=s_i+r_{s_i}$ and $\tilde{g_i}=g_i+r_{g_i}$ to be reconstructed without leakage by exchanging their $\ell$-bit $\langle \tilde{s_i} \rangle$ and $\langle \tilde{g_i} \rangle$. After that, $P_b$ finishes the computation by locally calculating $\langle g_i\cdot s_i \rangle_b =b \cdot \tilde{s_i}\cdot\tilde{g_i}-\tilde{s_i}\cdot\langle g_i \rangle_b-\tilde{g_i}\langle s_i \rangle \cdot \langle g_i\rangle_b+\langle r\rangle_b$.
However,  using $\ell$-bit $r_{g_i}$ to mask $g_i$ and $\ell$-bit $r_{s_i}$ to mask $s_i$ occur data inflation due to $g_i\in [-2^{\ell_f}, 2^{\ell_f}-1]$ and $s_i\in\{0,1\}$, which increases communication costs. To reduce communication overhead during the 2PC multiplication protocol, our main idea compresses intermediate $\langle \tilde{s_i} \rangle$ and $\langle\tilde{g_i} \rangle$ in a compact bit-width ring, while ensuring the output of $s_i\cdot g_i$ remains on the $\ell$-bit ring $\mathbb{Z}_{2^\ell}$.

\begin{algorithm}
\caption{Secure efficient 
 aggregation protocol $\Pi_{\text{Agg}}$}\footnotesize
\label{alg:agg}
\begin{algorithmic}
     \REQUIRE The secret shares $\langle \textbf{s}\rangle=\{\langle s_0\rangle,\ldots, \langle s_{N-1}\rangle\}$ and $\langle\textbf{g} \rangle=\{\langle g_0\rangle\ldots \langle g_{N-1}\rangle\}$.
    \ENSURE The secret share $\langle G \rangle=\sum_{i=0}^{N-1} \langle s_i \cdot g_i \rangle$
\end{algorithmic}
 \flushleft\underline{$\Pi_{\text{Agg}}$.offline($\ell$, $\ell_f$)}
\begin{algorithmic}[1]
    \FOR{$i\in [0,N-1]$ \textbf{in parallel}}
    \STATE $r_{s_i}\overset{\$}{\gets} \mathbb{Z}_{2}$, $r_{g_i}\overset{\$}{\gets} \mathbb{Z}_{2^{\ell'}}$, where $\ell' = \ell_f+2$.
    \STATE $u_i = r_{s_i}\cdot r_{g_i}$, $v_i =r_{g_i}\gg(\ell'-1)$, and  $ m_i =r_{s_i}\cdot v_i$ 
    \STATE $\langle r_{s_i}\rangle, \langle r_{g_i} \rangle, \langle u_i \rangle, \langle v_i\rangle, \langle m_i\rangle~\gets~\text{Share}(r_{s_i},r_{g_i},u_i, v_i,m_i)$
    \STATE $\mathcal{K}_b = \langle r_{s_i}\rangle_b\| \langle r_{g_i} \rangle_b\|\langle u_i \rangle_b\| \langle v_i\rangle_b\|\langle m_i\rangle_b$.
    \STATE Send $\mathcal{K}_b$ to the party $P_b$.
    \ENDFOR
\end{algorithmic}
   \flushleft   \underline{$\Pi_{\text{Agg}}$.online($\langle \textbf{s}\rangle$,$\langle \textbf{g}\rangle)$}
   \begin{algorithmic}[1]
   \setcounter{ALC@line}{7}
   \FOR{$i \in [0,N-1]$ \textbf{in parallel}}
    \STATE $\langle r_{s_i}\rangle_b, \langle r_{g_i} \rangle_b, \langle u_i \rangle_b, \langle v_i\rangle_b, \langle m_i\rangle_b \gets \mathcal{K}_b$
    \STATE $[\![s_i]\!] =\langle s_i \rangle \mod 2$
    \STATE $\langle \hat{g}'_i \rangle^{\ell'}= \langle g_i\rangle +\langle r_{g_i}\rangle \mod 2^{\ell'}$
    \STATE $[\![\hat{s_i}]\!] = [\![s_i]\!]+\langle r_{s_i}\rangle\mod 2$.
    \STATE  $P_b$ sends $\ell'$-bit $\langle \hat{g}'_i \rangle^{\ell'}_{b}$ and  $1$-bit $[\![\hat{s_i}]\!]_b$  to $P_{1-b}$.
    \label{alg:agg_line_1}
    \STATE $P_b$ receives  $\langle \hat{g}'_i \rangle^{\ell'}_{1-b}$ and  $[\![\hat{s}]\!]_{1-b}$ from $P_{1-b}$.
    \label{alg:agg_line_2}
    \STATE $\hat{g}'_i = \langle \hat{g}'_i \rangle^{\ell'}_{b}+\langle \hat{g}'_i \rangle^{\ell'}_{1-b} +2^{\ell'} \mod 2^{\ell'+1}$ and $\hat{s_i}= [\![\hat{s}]\!]_b\oplus [\![\hat{s}]\!]_{1-b}$.
    \STATE  $v_{g_i} =\neg\hat{g}'_i\gg (\ell'-1)$
    \STATE $\langle s_i \cdot g_i \rangle = \hat{s_i}\hat{g}'_i+\hat{s}v_{g_i}\langle v_i\rangle\cdot 2^{\ell'}+\hat{s}\langle r_{g_i}\rangle-\hat{s}\cdot2^{\ell'}+(1-2\hat{s_i})\hat{g}'_i\langle r_{s_i}\rangle+(1-2\hat{s}_i)v_{g_i}\langle m_{i}\rangle\cdot 2^{\ell'}+(1-2\hat{s}_i)(\langle u_i\rangle-\langle r_{s_i}\rangle\cdot2^{\ell'})$
    \ENDFOR
    \RETURN $\langle G \rangle=\sum_{i=0}^{N-1} \langle s_i \cdot g_i \rangle$.
   \end{algorithmic}
\end{algorithm}

 First, we explain how to compress intermediate messages $\tilde{s}_i$ and $\tilde{g}_i$ (ref. Lines 9-14 in Algorithm~\ref{alg:agg}). In our protocol, the arithmetic secret-shared $\langle s_i \rangle$ over the ring $\mathbb{Z}_{2^{\ell}}$ is scaled into the Boolean secret-shared $[\![s_i]\!]$  through a local modulo $2$ operation, i.e., $[\![s_i]\!]=\langle s_i\rangle \mod 2$. Then, each party computes $[\![\hat{s}_i]\!]=[\![s_i]\!]+\langle r_{s_i} \rangle \mod 2$ such that we can reconstruct a $1$-bit compressed $\hat{s_i}=s_i\oplus r_{s_i} \mod 2$ by exchanging the Boolean secret-shared $[\![\hat{s}_i]\!]$ without leakage, where $ r_{s_i} \in\{0,1\}$ is generated in the offline and is secret-shared between $P_0$ and $P_1$.  
For the arithmetic secret-shares of $g_i\in [-2^{\ell_f}, 2^{\ell_f}-1]$, only the sign and fractional parts of $g_i$ need to be masked with an $(\ell_f+1)$-bit random value $r_{g_i}\in [-2^{\ell_f}, 2^{\ell_f}-1]$. This allows us to scale $\langle \hat{g}_i\rangle$ into the $\ell'$-bit secret share $\langle \hat{g}_i\rangle^{\ell'}$ through a local modulo $2^{\ell'}$ operation, i.e., $\langle \hat{g}_i\rangle^{\ell'} =\langle \hat{g_i}\rangle+\langle r_{g_i}\rangle \mod 2^{\ell'}$, where $\ell'$ is set as $ \ell_f+2$ since $g_i+r_i \in [-2^{\ell_f+1},2^{\ell_f+1}-2]$. Thus, we can reconstruct an $\ell'$-bit compressed $g^{\ell'}_i=g_i +r_i\mod 2^{\ell'}$ by exchanging the secret-shared $\langle \hat{g}_i\rangle^{\ell'}$. 

What remains is to explain how to ensure the result of $s_i\cdot g_i$ still is secret-shared on the $\ell$-bit ring $\mathbb{Z}_{2^\ell}$ (ref. Lines 15-17 in Algorithm~\ref{alg:agg}). Since the compressed $\hat{s}_i$ and $\hat{g}_i$ are encoded using different bit widths,  they cannot be used directly for secure computation of  $s_i\cdot g_i$ over the ring $\mathbb{Z}_{2^\ell}$.
Thus, we need the upcast conversion from a different small bit-width fixed-point representation to a uniform  $\ell$-bit fixed-point representation. Fortunately, there is a useful relationship between  $\hat{s}_i$ and $s_i$ over the ring $\mathbb{Z}_{2^\ell}$, and similarly for $g_i$ and $ \hat{g_i}$. That is:
 \begin{equation}
\label{eq:mask-open}
    \begin{aligned}
    s_i~\text{mod}~2^{\ell}&= (\hat{s}_i+r_{s_i}-2r_{s_i}\hat{s}_i)~\text{mod}~2^{\ell}\\
    g_i~\text{mod}~2^{\ell}&= (\hat{g}_i+w\cdot2^{\ell'}-r_{g_i}-2^{\ell'}))~\text{mod}~2^{\ell}
    \end{aligned}
\end{equation}
where $w=(g_i+r_{g_i})\overset{?}{>}2^{\ell'}$. However, we can not directly compute $w$ since it involves non-linear comparisons that are more expensive than linear operations in MPC. To further facilitate the computation, we employ a positive heuristic trick used in Ditto\cite{wu2024ditto}. That is, we add a large bias $2^{\ell'-1}$ to $g_{i}$ to ensure that $g'_i = g_i+2^{\ell_f} \in [0,2^{\ell'-1}]$ is positive over the ring $\mathbb{Z}_{2^{\ell'}}$. The caculation $w=(g_i+r_{g_i})\overset{?}{>}2^{\ell'}$ then can be converted into $w=(r_{g_i} \gg(\ell'-1))\cdot \neg ((g'_i+r_{g_i}) \gg (\ell'-1))$, where $\gg$ denotes a right-shift operation, and $ \neg$ denotes a negation operation. After the upcast conversion,  the bias can be directly subtracted to eliminate its influence.  As a result,  the secure computation of $s_i\cdot g_i$ over the $\mathbb{Z}_{2^{\ell}}$ becomes:
\begin{equation}
\label{eq:s_i_mul_g_i}
  \begin{aligned}
     &s_i\cdot g_i~\text{mod}~2^{\ell} \\& = ((\hat{s_i}-r_{s_i}+2\hat{s_i}r_s)\cdot (\hat{g}'_i+w\cdot2^{\ell'}-r_{g_i}-2^{\ell'}))~\text{mod}~2^{\ell}\\
     &= \hat{s_i}\hat{g}'_i+\hat{s_i}w\cdot2^{\ell'}-\hat{s_i}(r_{g_i}-2^{\ell'})+(1-2\hat{s_i})\hat{g}'_ir_{s_i}\\
     &+(1-2\hat{s_i})r_{s_i}w\cdot 2^{\ell'}+(1-2\hat{s_i})(r_{s_i}r_{g_i}-2^{\ell'}r_{s_i})\\
     & = \textcolor{red}{\hat{s_i}\hat{g}'_i}+\textcolor{red}{\hat{s_i}}\textcolor{red}{v_{g_i}}\textcolor{blue}{v_i}\cdot2^{\ell'}+\textcolor{red}{\hat{s_i}}\textcolor{blue}{r_{g_i}}-\textcolor{red}{\hat{s_i}\cdot2^{\ell'}}+\textcolor{red}{(1-2\hat{s_i})\hat{g}'_i}\textcolor{blue}{r_{s_i}}\\
     &+\textcolor{red}{(1-2\hat{s_i})}\textcolor{red}{v_{g_i}\cdot2^{\ell'}}\textcolor{blue}{m_i}+\textcolor{red}{(1-2\hat{s_i})}\textcolor{blue}{(u_i-r_{s_i}\cdot2^{\ell'})}
  \end{aligned}
\end{equation}
where $u_i = r_{s_i}\cdot r_{g_i}$, $v_{r_{g_i}}=r_{g_i}\gg(\ell'-1)$,
$m_i = r_{s_i} \cdot (r_{g_i}\gg(\ell'-1))$ and $v_{g_i} =\neg\hat{g}'_i\gg (\ell'-1)$. In the equation, we use different colors to denote the nature of variables: \textbf{red} for input-dependent variables computed locally online and \textbf{blue} for input-independent random variables computed offline. These blue random variables are used to generate a pair of keys, $\mathcal{K}_{0}$ and $\mathcal{K}_{1}$ (ref. Lines 2-5 in the Algorithm~\ref{alg:agg}) in the offline phase. Each key $\mathcal{K}_{b}$ is then securely sent to the corresponding party. This pre-computation step ensures that when the online phase begins, the protocol can efficiently compute the necessary values with minimal communication overhead.

\subsection{Putting Everything Together}

\begin{algorithm}[hpt]
\footnotesize
\caption{Training a GDBT tree in Guard-GBDT}
\label{alg:training-FSS-DT}
\begin{algorithmic}
    \REQUIRE The vertical dataset $\textbf{X}_b\in \mathbb{R}^{N\times F_0}$; The sample space $\langle\textbf{s}\rangle$; The first-order gradients $\langle\textbf{g} \rangle$; The second-order gradients $\langle \textbf{h}\rangle$; The current node $node$ of $\mathcal{T}_{t,b}$; The tree depth $d=0$ of tree.
\ENSURE A distributed GBDT tree $\mathcal{T}_{t,b}$ for $P_b$

\renewcommand{\algorithmicrequire}{\textbf{Public hyperparameters:}}
\REQUIRE the bucket number $B$, the maximum tree depth $D$
\end{algorithmic}
\begin{algorithmic}[1]
\renewcommand{\algorithmicrequire}{\textbf{Function:}}
\REQUIRE $\text{SecureBuildTree}(\textbf{X}_b,[\![\textbf{s}]\!],\langle\textbf{g} \rangle,\langle \textbf{h}\rangle, \mathcal{T}_{t,b}.node, d)$:
\IF{ $d<D$ }
\STATE \label{alg:training:line-best-split} $(\langle z_{*} \rangle,\langle u_{*} \rangle)\gets\text{SecureBestSplit}($\textbf{X}$,\langle\textbf{s}\rangle,\langle\textbf{g}\rangle,\langle \textbf{h}\rangle)$
\STATE[\textbf{Open best-split identifier:}] \label{alg:training:line-best-split- identifier}  $c= \text{Open}(\langle z_{*} \rangle -F_0 <0)$; $P_c$ receives ($\langle z_{*}\rangle_{1-c}$, $\langle u_{*}\rangle_{1-c}$) from $P_{1-c}$; $P_c$ opens  $z_{*}= \langle z_{*}\rangle_{1-c}+\langle z_{*}\rangle_{c}$ and $u_* = \langle u_{*}\rangle_{c}+\langle u_{*}\rangle_{1-c}$.
\STATE[\textbf{Update left and right sample space:}]\label{alg:training:line-sample space}$P_c$ set $\langle \textbf{s}_{test}\rangle_c=\{\textbf{X}_{c}[1,z_*]<\textbf{Bin}_c[z_*,u_*],\ldots,\textbf{X}_{c}[N,z_*]<\textbf{Bin}_c[z_*,u_*]\}$ and $\langle\textbf{s}_{test}\rangle_{1-c}=\textbf{0}$.
\STATE $P_0$ and  $P_1$  compute  $\langle \textbf{s}_{L^*}\rangle= \langle\textbf{s}_{test}\rangle\cdot\langle\textbf{s}\rangle, \langle \textbf{s}_{R^*}\rangle=(\textbf{1}-\langle\textbf{s}_{test}\rangle)\cdot \langle \textbf{s}\rangle$

\STATE [\textbf{Recods $z_*$ and $u_*$ into tree $\mathcal{T}_{t,b}$:}] $P_c$ records $\mathcal{T}_{t,b}.\text{node}.value = (z^{(k)}_*, u^{(k)}_*)$; $P_{1-c}$ records $\mathcal{T}_{t,b}.\text{node}.value =(-1,-1)$
\STATE [\textbf{Build subtrees:}] $\text{SecureBuildTree}($\textbf{X}$,\langle \textbf{s}_{L^*}\rangle,\langle\textbf{g} \rangle,\langle \textbf{h}\rangle, \mathcal{T}_{t,b}.\text{left},d+1)$; $\text{SecureBuildTree}($\textbf{X}$,\langle \textbf{s}_{R^*}]\rangle,\langle\textbf{g} \rangle,\langle \textbf{h} \rangle, \mathcal{T}_{t,b}.\text{right},d+1)$
\ELSE
\STATE \flushleft[\textbf{Build leaf node:}]
$\langle G \rangle = \Pi_{\text{Agg}}.\text{online}(\langle \textbf{s}\rangle, \langle \textbf{g}\rangle)$; $\langle H \rangle = \Pi_{\text{Agg}}.\text{online}(\langle \textbf{s}\rangle, \langle \textbf{h}\rangle)$; $\langle w \rangle = \Pi_{\text{LUT}_{w}}.\text{online}(\langle G\rangle,\langle H\rangle)$; $\mathcal{T}_{t,b}.node.value =(\langle w\rangle_b,-1)$;
\ENDIF
\RETURN $\mathcal{T}_{t,b}$ 
\end{algorithmic}
\begin{algorithmic}[1]
    \renewcommand{\algorithmicrequire}{\textbf{Function:}}
\REQUIRE $\text{SecureBestSplit}($\textbf{X}$,\langle\textbf{s}\rangle,\langle\textbf{g} \rangle,\langle \textbf{h}\rangle)$:
\setcounter{ALC@line}{11}
\FOR{$z\in \{1,\ldots, F\}$}
\FOR{ $u \in \{1,B-1\}$ \textbf{in parallel}}
\STATE [\textbf{Compute sample space:}]
$P_b$ computes $\langle \textbf{s}_{test}\rangle_b=\{\textbf{X}_{b}[1,z]<\textbf{Bucket}_b[z, u],\ldots,\textbf{X}_{b}[N,z]<\textbf{Bucket}_b[z,u]\}$ and $P_{1-b}$ sets $\langle \textbf{s}_{test}\rangle_{1-b}=\textbf{0}$;$\langle \textbf{s}_L\rangle=\langle \textbf{s}\rangle\cdot\langle\textbf{s}_{test}\rangle$ and $\langle\textbf{s}_R\rangle=\langle\textbf{s}_{test}\rangle\cdot(\textbf{1}-\langle\textbf{s}\rangle)$.
\STATE [\textbf{Aggregate gradients:}]
$\langle G_{L}\rangle=\Pi_{\text{Agg}}.\text{online}(\langle\textbf{s}_L\rangle,\langle \textbf{g} \rangle)$;  $\langle G_{R}\rangle=\Pi_{\text{Agg}}.\text{online}(\langle\textbf{s}_R\rangle,\langle \textbf{g} \rangle)$;
$\langle H_{L}\rangle$$=\Pi_{\text{Agg}}.\text{online}(\langle\textbf{s}_L\rangle,\langle \textbf{h} \rangle)$;
$\langle H_{R}\rangle=\Pi_{\text{Agg}}.\text{online}(\langle\textbf{s}_R\rangle,\langle \textbf{h} \rangle)$.
\STATE [\textbf{Compute gain:}] $\langle \textbf{Sign}\rangle=\langle 2H_L-H_X\rangle<0$; $\langle \mathcal{G}^{(z,u)} \rangle=\langle H_R+\gamma \rangle\cdot\langle (G_L)^2\rangle+\langle H_L+\gamma\rangle \cdot\langle (G_R)^2\rangle$; $\langle \mathcal{G}^{(z,u)} \rangle=\langle \textbf{Sign}\rangle\cdot\langle \mathcal{G}^{(z,u)} \rangle+\langle 1-\textbf{Sign}\rangle\cdot\langle \mathcal{G}^{(z,u)} \rangle$;
\ENDFOR
\ENDFOR
\STATE [\textbf{The best split:}] $\langle z_* \rangle,\langle u_* \rangle=\Pi_{\text{Argmax}}(\{\langle \mathcal{G}^{(1,1)} \rangle,\ldots,\langle \mathcal{G}^{(F,B-1)} \rangle\})$
\RETURN $(\langle z_* \rangle,\langle u_* \rangle)$
\end{algorithmic}
\end{algorithm}

Here we introduce how Guard-GBDT combines the above components to train a tree in the GBDT model over the vertical datasets $\textbf{X}_0$ and $\textbf{X}_1$.  Guard-GBDT first invokes the offline program of our components to preprocess all randomness independent of the online inputs, and then trains a GBDT model, which is given in Algorithm \ref{alg:training-FSS-DT}. 
Each party $P_b$ inputs its dataset $\textbf{X}_b\in \mathbb{R}^{N\times F_0}$, secret-shared sample space $\langle \textbf{s}\rangle$, the secret-shared first-order gradients $\langle\textbf{g} \rangle$, the secret-shared second first-order gradients $\langle \textbf{h}\rangle$ and outputs a distributed GBDT tree $\mathcal{T}_{t,b}$. We assume that the first-order and second-order gradients have been computed privately and given as inputs in the algorithm. This allows us to directly reuse the algorithm to build the next tree. For a binary classification task using the cross-entropy loss, our secure $\Pi_{\text{LUT}_{\delta}^{n}}$  protocol in Algorithm \ref{alg:secure-sigmoid-protocol} can be used to compute $\langle \textbf{g} \rangle=\Pi_{\text{LUT}_{\delta}^{n}}.\text{online}(\langle \hat{\textbf{y}}^{(t-1)} \rangle)-\textbf{y}$ and $\langle \textbf{h} \rangle=\Pi_{\text{LUT}_{\delta}^{n}}.\text{online}(\langle \hat{\textbf{y}}^{(t-1)} \rangle)\cdot \Pi_{\text{LUT}_{\delta}^{n}}.\text{online}(\langle \hat{\textbf{y}}^{(t-1)} \rangle)$, where the initial prediction is $\hat{\textbf{y}}=\langle\textbf{0}\rangle$ for the first tree. 
 
 With the help of the gradients, the algorithm first uses a secure best-split function (Line \ref{alg:training:line-best-split} in  Algorithm \ref{alg:training-FSS-DT}) to determine the best-split identifier ($z_*, u_*$). In the function,  each party $P_b$ first picks all possible thresholds $\textbf{Bucket}_b$ of
each feature locally from its own dataset $\textbf{X}_b$. For a possible split
candidate $\textbf{Bucket}_b[z,u]$, party $P_b$ generates a local test of sample space $\langle\textbf{s}_{\text{test}}\rangle_b=\{\textbf{X}[0,z]<\textbf{Bucket}_b[z,u],\ldots,\textbf{X}[N,z]<\textbf{Bucket}_b[z,u]\}$, while the party $P_{1-b}$ sets $\langle \textbf{s}_{{test}}\rangle_{1-b}=\textbf{0}$. Next, the two parties compute the corresponding possible left sample space $\langle\textbf{s}_L\rangle=\langle\textbf{s}_{test}\rangle \cdot\langle\textbf{s}\rangle$ and the right sample space $\langle \textbf{s}_R\rangle=(1-\langle \textbf{s}_{test}\rangle)\cdot\langle \textbf{s}\rangle$. We use our secure $\Pi_{\text{Agg}}$ protocol to aggregate the gradients on the left and right sample space like $\langle G_L \rangle$, $\langle G_R \rangle$, $\langle H_L\rangle$, and $\langle H_R\rangle$.  After that, the two parties can compute the splitting 
scores \{$\langle \mathcal{G}^{z,u}\rangle$\} jointly from $\{\langle G_L\rangle,\langle G_R\rangle, \langle H_L\rangle, \langle H_R\rangle \}$ according to our split metric, secure multiplication, and secure DCF-based comparison. Finally, the best-split identifier ($\langle z_*\rangle, \langle u_*\rangle$) is determined using the  $\Pi_{\text{Argmax}}$ protocol of SiGBDT.

Once the best-split identifier ($\langle z_*\rangle, \langle u_*\rangle$) is computed, we can only reveal it to its holder (Line \ref{alg:training:line-best-split- identifier} in Algorithm \ref{alg:training-FSS-DT}) according to the secure definition and the existing private GDBT training paradigm. To this end, we compute a 1-bit indicator $c=\text{Open}{\langle z_* \rangle<F_0}$ to indicate the chosen ($z_*,u_*$) belongs to the $P_c$. After that, $P_c$ records ($z_*,u_*$) into the current node of $\mathcal{T}_{t,b}$ and $P_{1-c}$  records (-1,-1) into the current node of $\mathcal{T}_{t,b}$. Next, $P_c$ can locally test $ts = \{\textbf{X}_{c}[1,z_*]<\textbf{Bin}_c[z_*,u_*],\ldots,\textbf{X}_{c}[N,z_*]<\textbf{Bin}_c[z_*,u_*]\}$ and set $\langle \textbf{s}_{test}\rangle_c=ts$ and $\langle\textbf{s}_{test}\rangle_{1-c}=\textbf{0}$. Then, the two parties can update the left sample space $\langle \textbf{s}_{L^*}\rangle=\langle \textbf{s}\rangle \cdot\langle \textbf{s}_{test}\rangle$ and the right sample space $\langle \textbf{s}_{R^*}\rangle=\langle \textbf{s}\rangle \cdot(1-\langle \textbf{s}_{test}\rangle$ (Line \ref{alg:training:line-sample space} in Algorithm \ref{alg:training-FSS-DT}). After that, $P_0$ and $P_1$ loop through the above process to train the left and right child of the current node. When reaching the maximum depth, $P_0$ and $P_1$ aggregates gradients in the sample space of the current leaf to compute leaf weight $\langle w \rangle$ using secure  $\Pi_{\text{LUT}_{w}^{n}}$ protocol. Finally, the secret-shared leaf weight $\langle w \rangle$ is recorded into the current node of $\mathcal{T}_{t,b}$.

\noindent\textbf{Secure prediction.} Once the GBDT model is trained, secure prediction can be performed using an oblivious algorithm. Previous works, such as Pivot~\cite{wu2020privacy} and SecureBoost\cite{cheng2021secureboost}, have applied HE for secure GBDT prediction in the MPC model. In contrast, we adopt a similar approach to SiGBDT\cite{jiang2024sigbdt} and Squirrel~\cite{lu2023squirrel} by employing the secret-sharing-based prediction algorithm proposed by HEP-XGB~\cite{fang2021large}, which avoids the time-consuming operations associated with HE. Briefly, for a new input sample, each party can locally prepare a binary vector $\textbf{path}_b =\textbf{0}$ for each tree. Recall that $P_b$ holds the tree part $\mathcal{T}_{t,b}$ which includes its features. The elements of $\textbf{path}_b$ are updated by comparing its features and corresponding node thresholds. That is, $\textbf{path}_b[k]=1$ indicates that the sample might be classified to the k-th leaf (the leaves are ordered from left to right), while $\textbf{path}_b[k]=0$ means the input will not be classified to the $k$-th leaf based on the split identifiers held by  $P_b$. At the end, there is only one entry of `1' in $\textbf{path}_0\cdot \textbf{path}_1$. We define $\langle \textbf{pth} \rangle_0 =\textbf{path}_0, \langle \textbf{pth}' \rangle_0= \textbf{0}$ held by $P_0$ and $\langle \textbf{pth} \rangle_1= \textbf{0}, \langle \textbf{pth}' \rangle_{1} =\textbf{path}_1$ held by $P_1$. The final prediction on the tree is computed as $\sum_k \langle \textbf{pth} \rangle[k]\cdot \langle \textbf{pth}' \rangle[k]\cdot \langle \textbf{w}\rangle[k]$  given the shares of the leaf weights.

\subsection{Security Analysis}
Due to the page limit, we provide a concise security analysis. A long formal security analysis will be provided in Appendix~\ref{appendix:security}. 
Guard-GBDT framework consists of multiple sub-protocols for smaller private computations and is secure against semi-honest PPT adversaries as defined in Definition~\ref{PPT-definition}.
During the training, all the data communicated between the two parties is just secret shares. Our linear operations including addition and multiplication are all performed securely in ASS and their security has been proved in the existing two-party computation framework, ABY~\cite{demmler2015aby}. The nonlinear comparison protocol is implemented with DCF, which does not leak either the input from the other party or the output and has been proved in the FSS framework~\cite{boyle2016function}. 
Notably, our \(\Pi_{\text{agg}}\) protocol achieves the same security and accuracy as existing MPC-based protocols. When the attacker knows the gradient lies in \([-1,1]\), the probability of correctly guessing the private data is \( \frac{1}{2^{\ell}} \Big/ \frac{1}{2^{\ell - \ell_f - 1}} = \frac{1}{2^{\ell_f+1}} \) in existing MPC-based works. Since our protocol protects only the sign bit and fraction, the attacker's success probability is also \( \frac{1}{2^{\ell_f}} \times \frac{1}{2} = \frac{1}{2^{\ell_f+1}} \). 
As a result, the adversary who controls one party cannot learn any data about the other party following the composition theorem~\cite{canetti2000security} of the semi-honest model. 

\subsection{Extension of Guard-GBDT}
\label{subsec:extension}
The Guard-GBDT framework can be flexibly applied to collaborative learning scenarios involving three or more parties, such as cross-enterprise collaborative learning and federated learning~\cite{fereidooni2021safelearn}. For the proposed approximation protocol and division-free split gain metric, the linear operations only need to replace the secret-sharing primitives that support three or more parties. The nonlinear operations can use the existing works of Boyle ~\cite{boyle2021function} to replace the two-party FSS.
{In an $N$-party setting,  each party $P_i$ receives a key $\mathcal{K}_i$ for a function $f$. When party $P_i$ computes $f_i(x)$ using $\mathcal{K}_i$ and their share of $x$, the sum of their outputs $\sum f_i(x)$ reveals $f(x)$ (or a secret sharing of $f(x)$).} 
For the proposed aggregation protocol, we use a 1-out-of-m secret sharing protocol to construct a multi-party protocol and use the idea in Section \ref{subsec:agg} to compress intermediate mask values for lower communication overhead. 

%% file: 6-expriments.tex
\section{Expriments Evaluation}
\subsection{Expriment Setup}
\noindent \textbf{Testbed Environment.}
We implement a prototype of Guard-GBDT using PyTorch framework~\cite{paszke2019pytorch} and evaluate it on a computer equipped with an  AMD Ryzen9~5900X CPU @ 3.20GHz, 512GB RAM, operating Ubuntu 20.
Following prior works~\cite{jiang2024sigbdt,bai2023mostree}, we utilize a Linux network tool, ``tc", to simulate the local-area network (LAN, RTT: 0.2 ms,
1 Gbps) and the wide-area network (WAN, RTT: 40ms, 100
Mbps) between the two parties on the same workstation.
Similarly to SiGBDT\cite{jiang2024sigbdt} and HEP-XGB~\cite{fang2021large}, all secure protocols are on the $\ell=64$ bit length ring. We set the fraction precision $\ell_f=16$ and the security parameter of FSS $\lambda=128$.

\noindent \textbf{Baseline.}
We compare Guard-GBDT with XGboost (plaintext) model to verify the accuracy of our work. We also compare Guard-GBDT with the state-of-the-art works: SiGBDT\cite{jiang2024sigbdt} and HEP-XGB\cite{fang2021large}, and implement their works following the details provided in their paper. We do not compare with other works (such as Squirrel\cite{lu2023squirrel}, Sgboost\cite{zhao2022sgboost}, NodeGuard\cite{dai2024nodeguard}, Privet\cite{zheng2023privet}, and SecureBoost\cite{cheng2021secureboost}), as SiGBDT and HEP-XGB have already outperformed them.
\begin{table}[h]
\centering
\footnotesize
\caption{Information of training datasets}
\begin{tabular}{lcccc}
\toprule 
 \textbf{Dataset} & \textbf{Instances} & \textbf{Features}\\
 \midrule 
Breast Cancer & 699& 9 \\
Credit &45211 &16 \\
phishing website &11055&67\\
Skin & 245057 & 3 \\
Covertype & 581012 & 54 \\
 \bottomrule 
\end{tabular}
\label{tab:datasets}
\end{table}

\noindent \textbf{Dataset.}
Following previous works~\cite{lu2023squirrel,fang2021large,jiang2024sigbdt,fu2019experimental}, we use 5 real-world datasets from UCI Machine Learning Repository\footnote{\texttt{\url{https://archive.ics.uci.edu/}}} to evaluate the accuracy and efficiency of Guard-GBDT. The datasets are summarized in Table~\ref{tab:datasets}. We follow a ratio of 8:2 for training and testing datasets for every dataset. Each training dataset is split vertically and evenly for $P_0$ and $P_1$. We assume that the samples in each party’s database have been properly aligned beforehand.
{Similar to SiGBDT and HEP-XGB, each dataset is preprocessed to discretize features into different bins via the existing equal-width binning method in the offline phase, where the max bin size is $B=8$.}

\subsection{Accuracy Evaluation}
\label{subsec:acc-exp}

Guard-GBDT employs lookup tables with multiple segment approximations to replace the sigmoid function and leaf weight computation. How the data are segmented could affect the model's accuracy. We first test model accuracy with different segment sizes on five real datasets to find the best segment size of approximations, where the total number of trees in GBDT is $T=5$ (Fig~\ref{fig:accuracy_under_segment_size_10} also presents the results for \(T = 10\), with almost identical conclusions).
As shown in Fig.~\ref{fig:accuracy_under_segment_size}, the experimental results for large datasets, such as covertype and skin, demonstrate stable performance with accuracy close to 100\% across various segment sizes ($n$) and tree depths ($D$), indicating that these datasets are insensitive to changes in segment size.
In contrast, {results on smaller datasets exhibit fluctuations and are more sensitive to segment size changes. These fluctuations became more pronounced as the tree depth increased. However, as the segment size grows, the accuracy at different depths tends to stabilize, suggesting that once the segment size exceeds a certain threshold (i.e., $ n \geq10$ for $D=4~\text{or}~8$ and $ n \geq12$ for $D=16~\text{or}~32$ ), the model's accuracy becomes consistent across different datasets}. 

\begin{figure}[ht]
    \centering

    \subfloat[\footnotesize{Tree depth:} D = 4]{
        \includegraphics[width=0.45\linewidth]{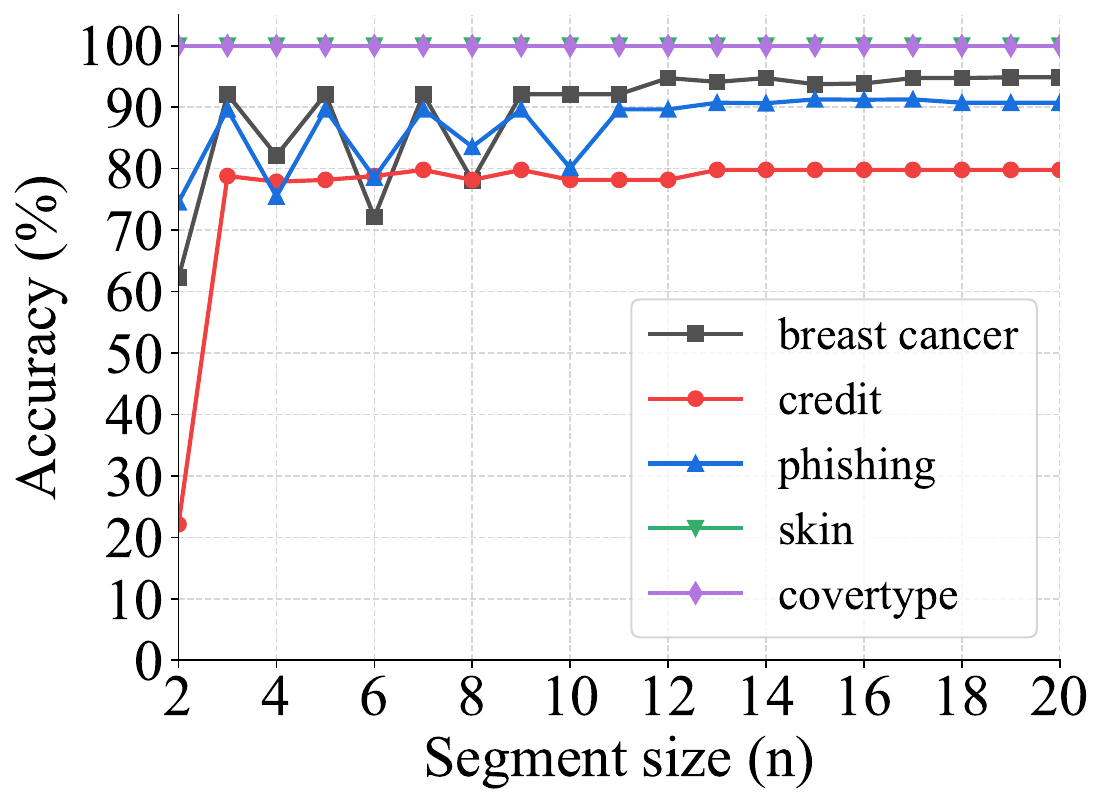}
        \label{fig:seg_depth_4}}
    \hfill
    \subfloat[\footnotesize{Tree depth:} D = 8]{
        \includegraphics[width=0.45\linewidth]{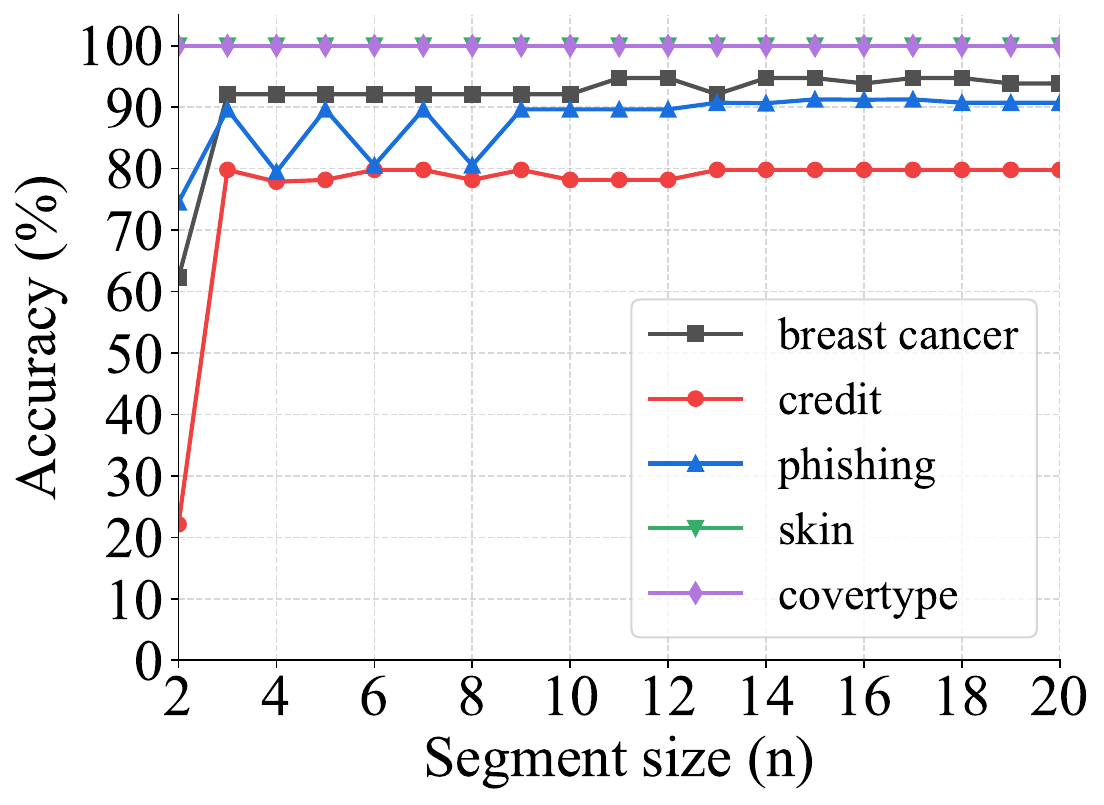}
        \label{fig:seg_depth_8}
    }   \hfill
     \subfloat[\footnotesize{Tree depth:} D = 16]{
        \includegraphics[width=0.45\linewidth]{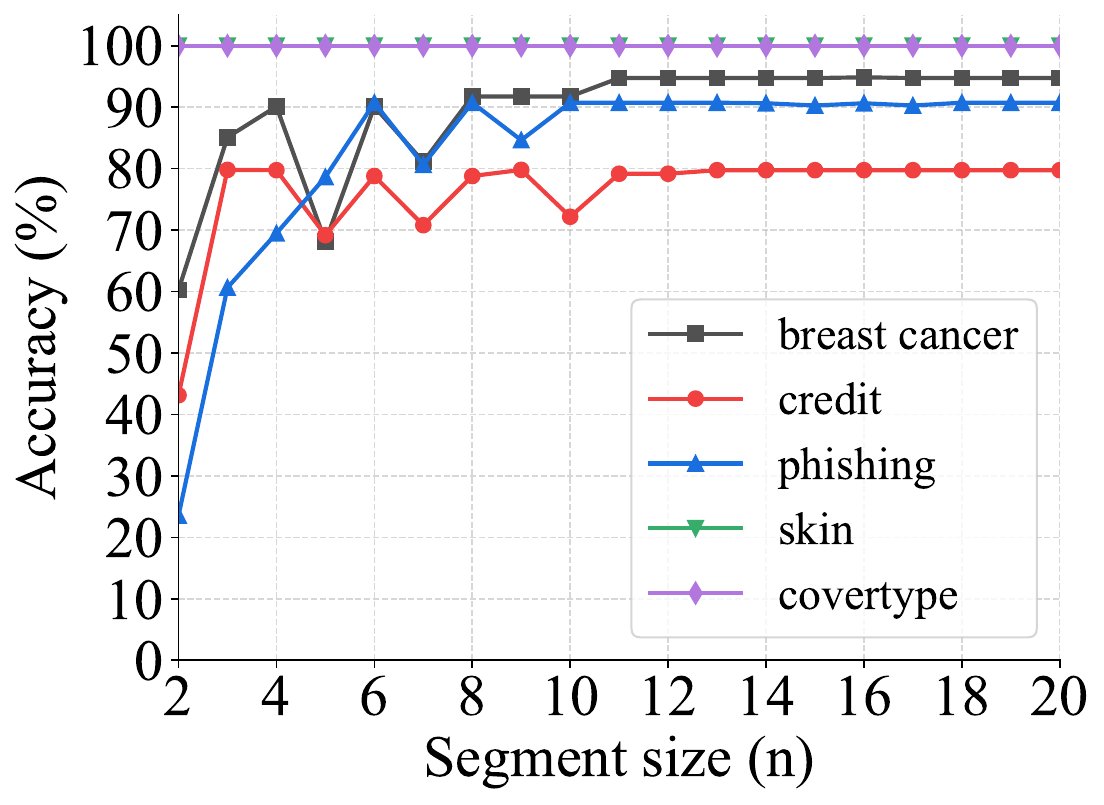}
        \label{fig:seg_depth_8}
    }   \hfill
     \subfloat[$\text{\footnotesize{Tree depth:} D = 32}$]{
        \includegraphics[width=0.45\linewidth]{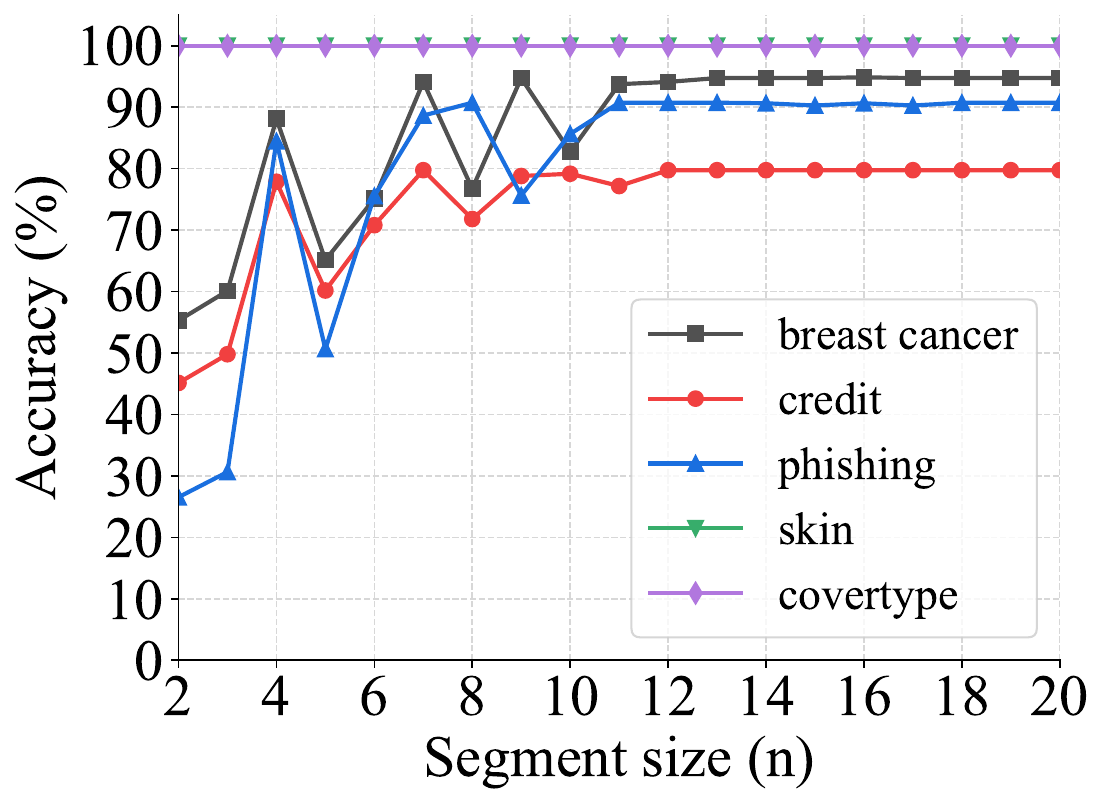}
        \label{fig:seg_depth_8}
    }
    \caption{Tree accuracy with different segment approximations.}
    \label{fig:accuracy_under_segment_size}
\end{figure}

\begin{figure}[h]
    \subfloat[\footnotesize{Tree depth:}  D = 4]{
        \includegraphics[width=0.45\linewidth]{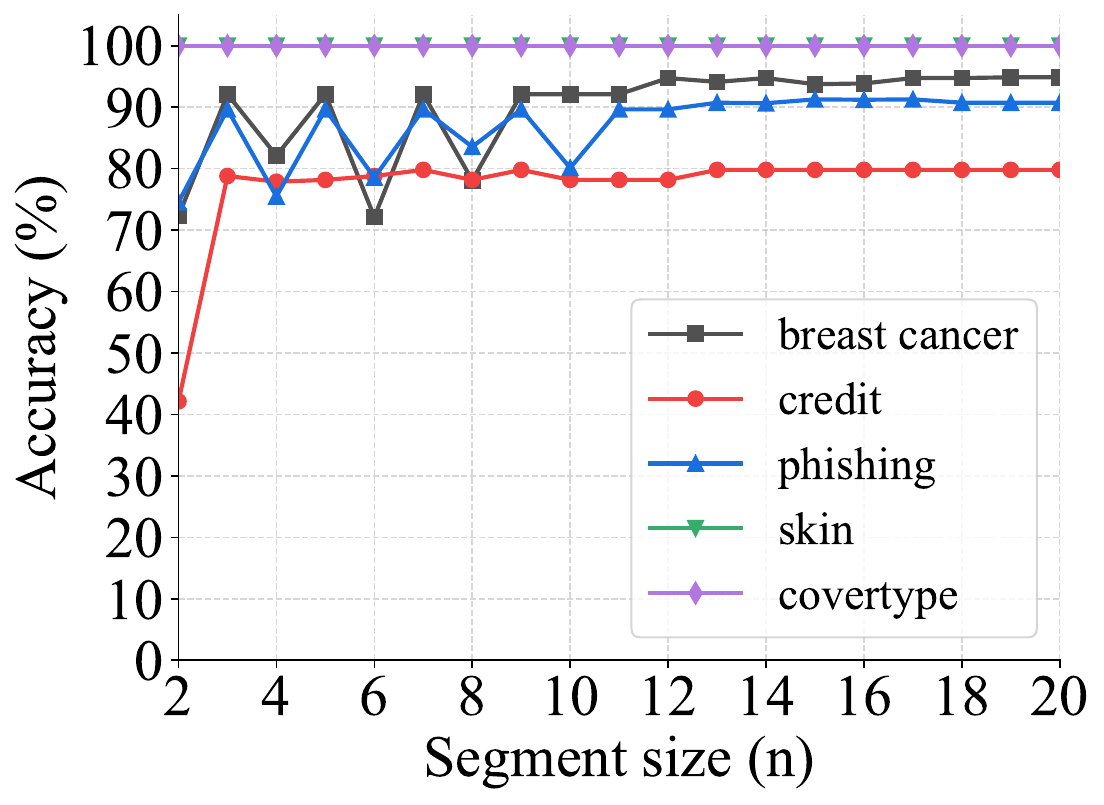}
        \label{fig:seg_depth_4}
    }    \hfill
    \subfloat[\footnotesize{Tree depth:} D = 8]{
        \includegraphics[width=0.45\linewidth]{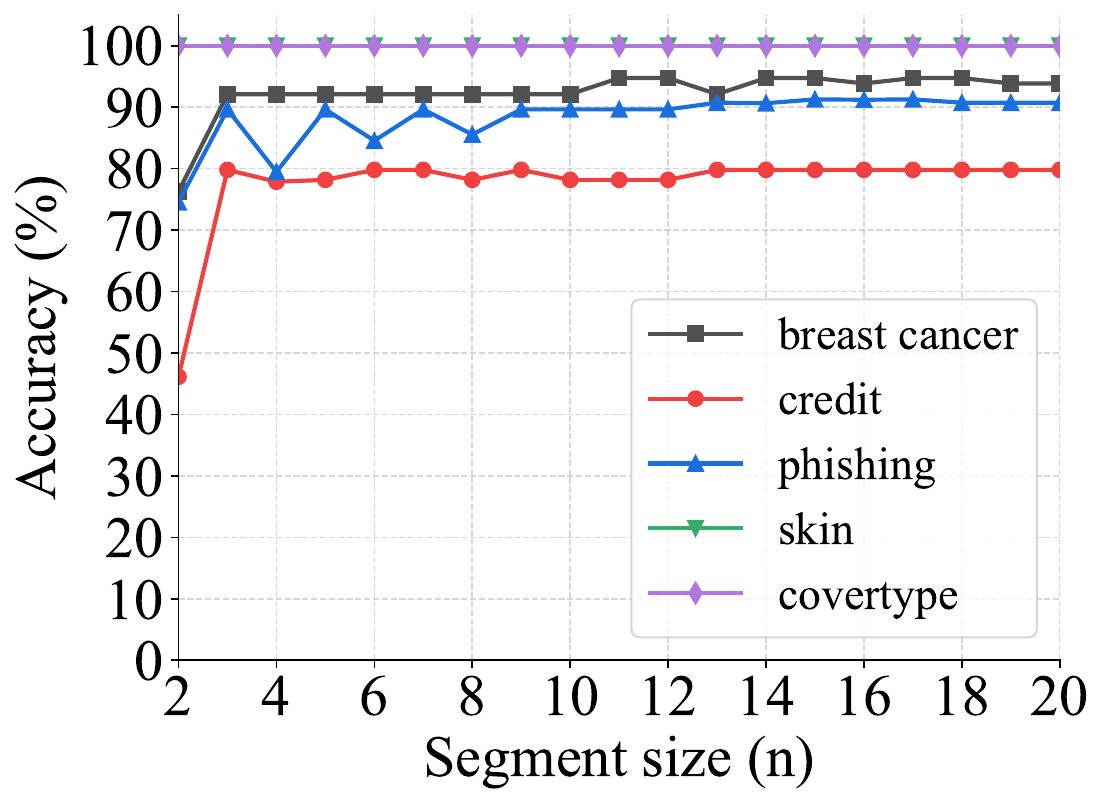}
        \label{fig:seg_depth_8}
    }     \hfill
     \subfloat[\footnotesize{Tree depth:} D = 16]{
        \includegraphics[width=0.45\linewidth]{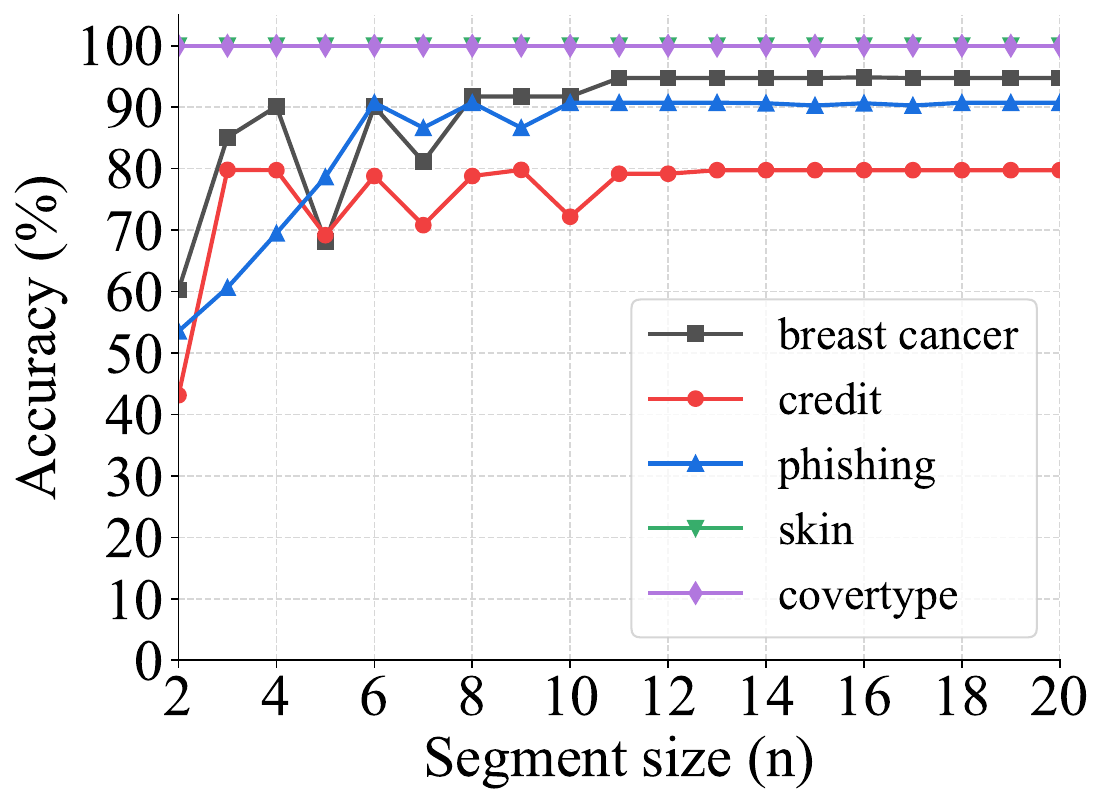}
        \label{fig:seg_depth_8}
    }     \hfill
     \subfloat[\footnotesize{Tree depth:} D = 32]{
        \includegraphics[width=0.45\linewidth]{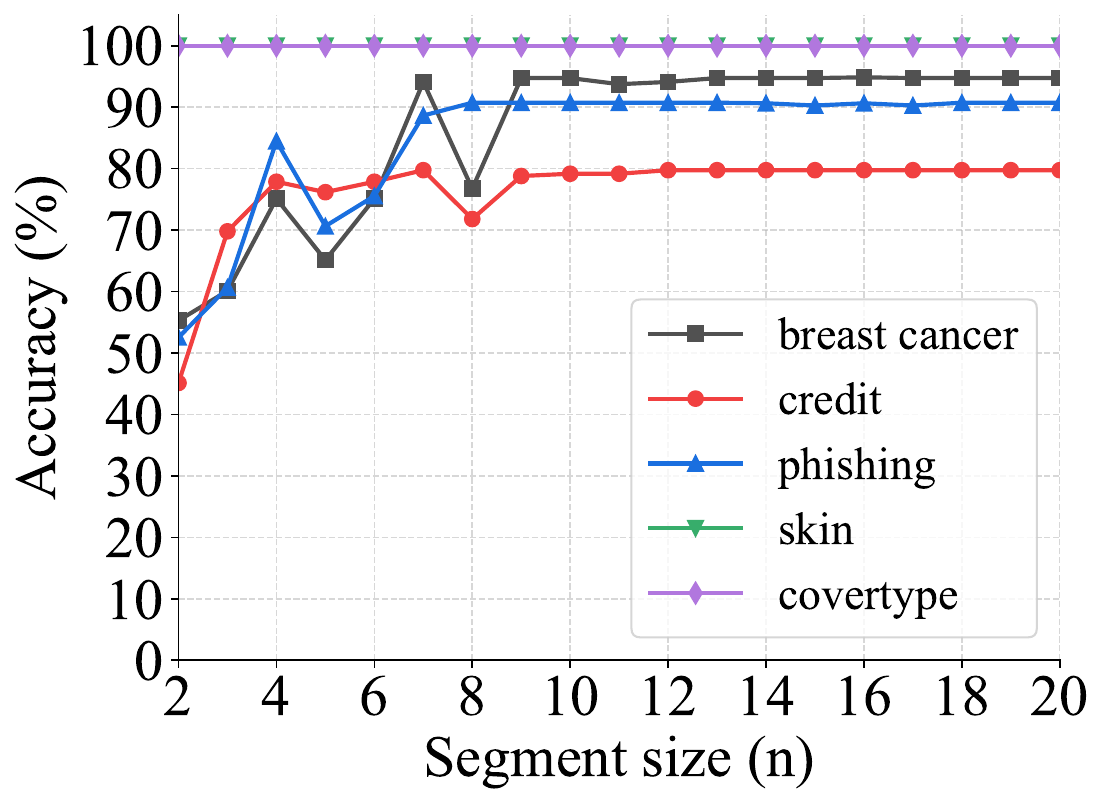}
        \label{fig:seg_depth_8}
    }
    \caption{Tree accuracy with different segment approximations, where $T=10$.}
    \label{fig:accuracy_under_segment_size_10}
\end{figure}

\begin{table}[h]
\scriptsize
\caption{Model accuracy over different tree depths}
\resizebox{1\linewidth}{!}{
\begin{tabular}{c c |ccc|c}
\toprule
\textbf{Depth} & \textbf{Dataset}       & \textbf{Guard-GBDT} & \textbf{SiGBDT} &\textbf{ HEP-XGB} & \textbf{XGBoost (Plain)} \\
\midrule
\multirow{4}{*}{D=4} &breast-cancer & \textbf{94.74}    & 93.86  & 87.72   & 96.37           \\
&credit        & 79.90    & \textbf{81.12}  &  76.34     & 82.42           \\
&phishing      & \textbf{89.42}    & 89.19  &   78.11    & 89.91           \\
&skin          & \textbf{100.00}      & 100.00    & 91,11    & 100.00       \\
&covertype     & \textbf{100.00}      & 100.00    & 90.81    & 100.00     \\
\hline
\multirow{4}{*}{D=8} &breast-cancer & \textbf{95.78}    & 92.86  & 90.72   & 97.37           \\
&credit        & 81.02    & \textbf{81.12}  &  77.14     & 82.37          \\
&phishing      & \textbf{90.00}    & 89.91  &  81.12    & 91.91           \\
&skin          & \textbf{100.00}      & 100.00    & 91,11    & 100.00       \\
&covertype     & \textbf{100.00}      & 100.00    & 90.81    & 100.00     \\
\hline
\multirow{4}{*}{D=16} &breast-cancer & \textbf{94.14}    & 92.86  & 89.72   & 97.37           \\
&credit        & 80.75    & \textbf{81.12}  &  77.14     & 82.15          \\
&phishing      & \textbf{90.00}  & 89.91  &   84.12    & 92.11           \\
&skin          & \textbf{100.00}      & 100.00    & 91.11    & 100.00       \\
&covertype     & \textbf{100.00}      & 100.00    & 90.81    & 100.00     \\
\hline
\multirow{4}{*}{D=32} &breast-cancer & \textbf{92.74}    & 92.86  & 89.12   & 97.37           \\
&credit        & 81.10    & \textbf{81.12}  &  77.90     & 82.17           \\
&phishing      & \textbf{90.00}    & 89.91  &   84.12    & 91.91          \\
&skin          & \textbf{100.00}      & 100.00    & 91,11    & 100.00       \\
&covertype     & \textbf{100.00}      & 100.00    & 90.81    & 100.00     \\
\bottomrule
\end{tabular}}
\label{tab:accuracy}
\end{table}

\begin{figure*}[ht]
  \centering
  \subfloat{\includegraphics[width=0.33\textwidth]{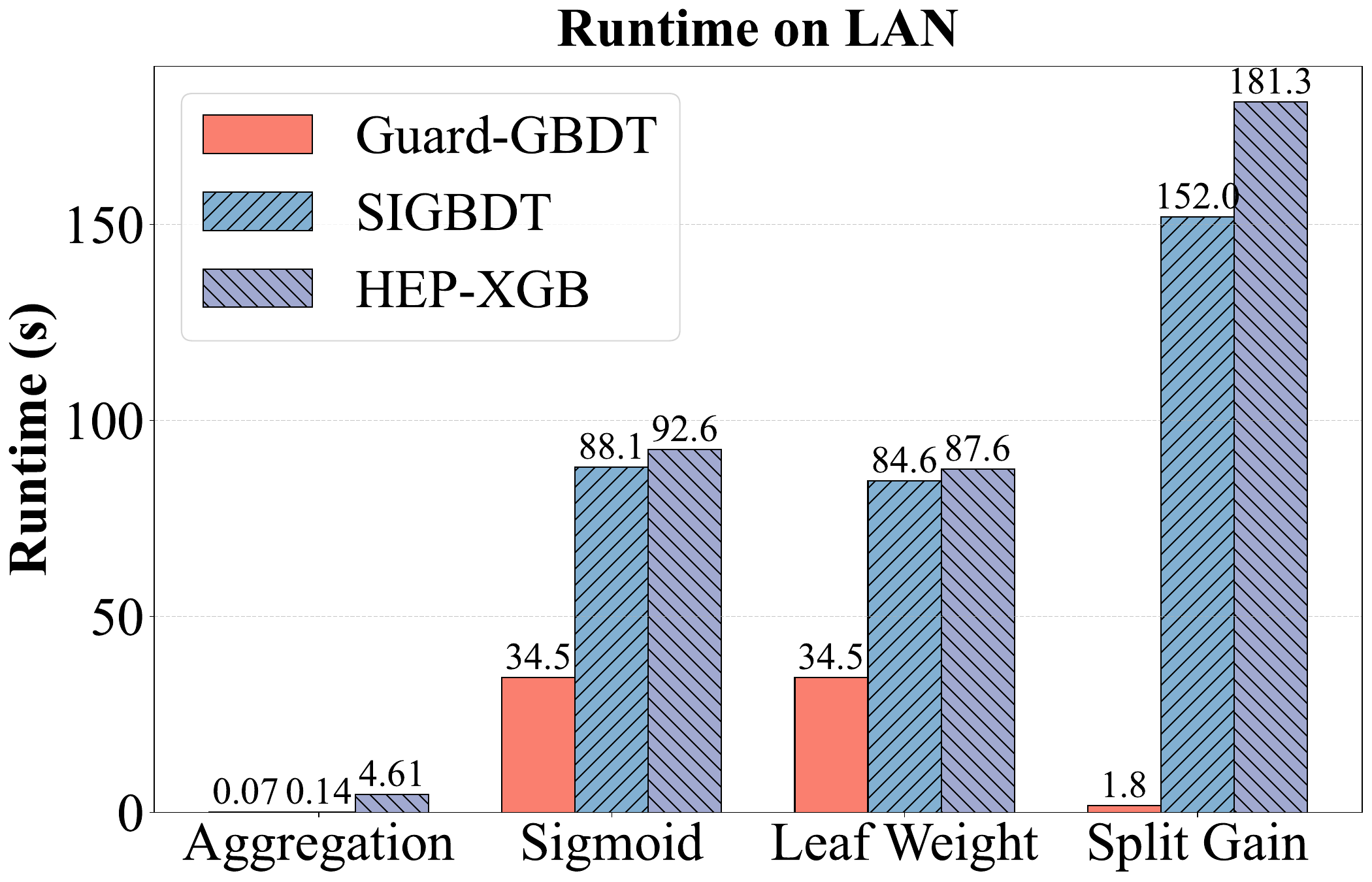}}
  \hfill
  \subfloat{\includegraphics[width=0.33\textwidth]{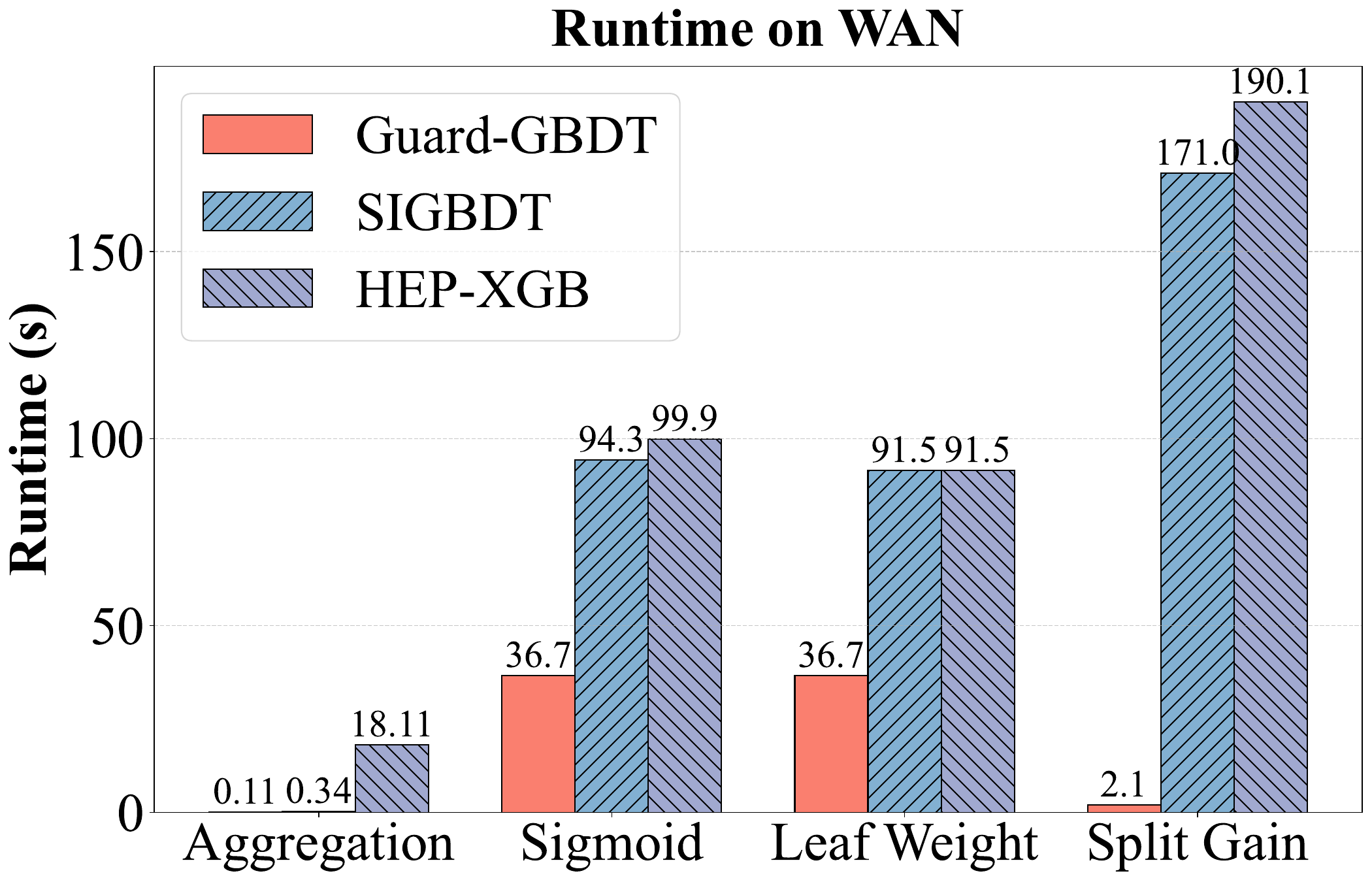}}
  \hfill
  \subfloat{\includegraphics[width=0.33\textwidth]{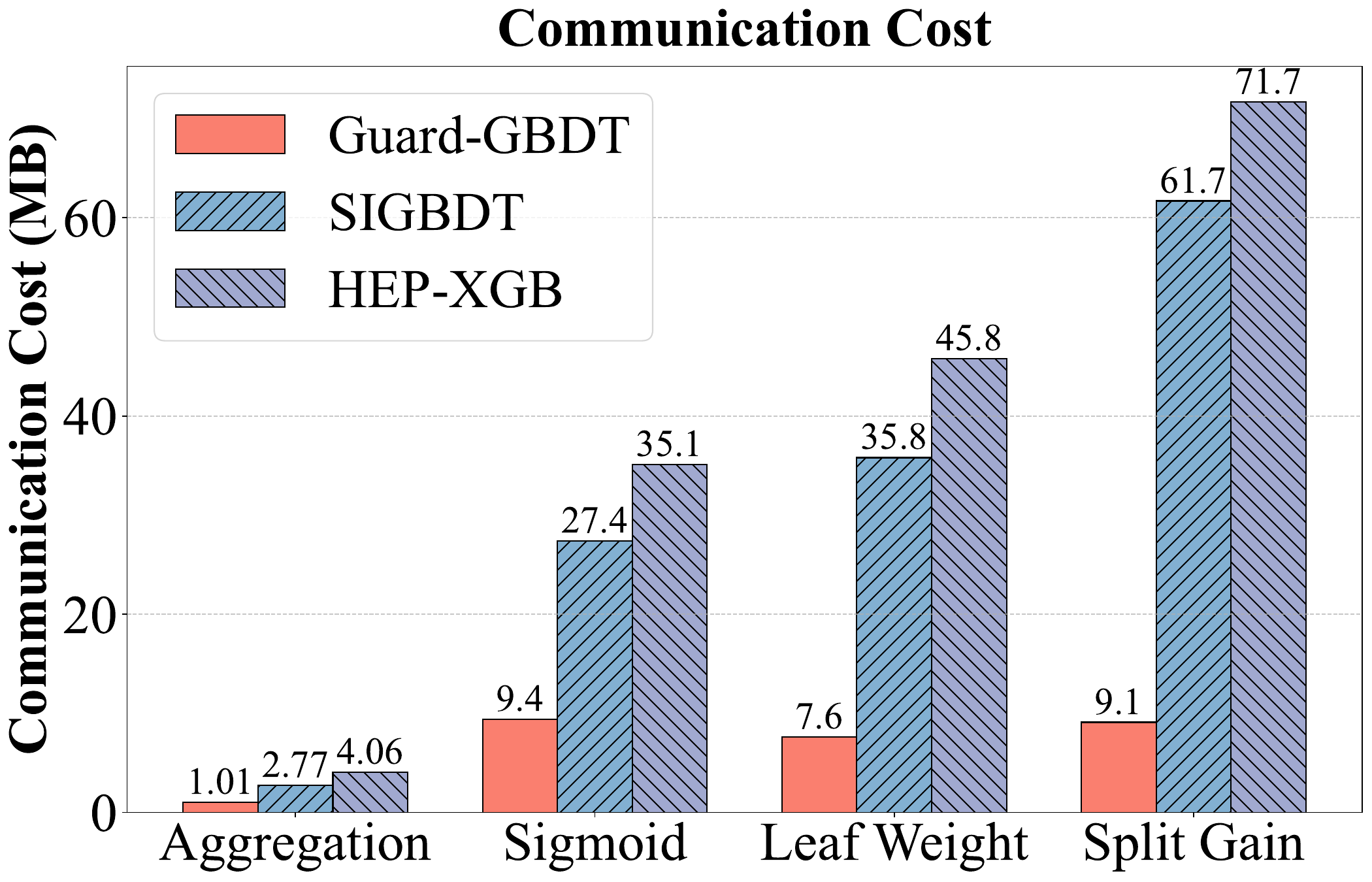}}
  \caption{Microbenchmarks with input size of $10^5$}
  \label{fig:me}
\end{figure*}

\begin{figure}[h]
    \centering
    \includegraphics[width=0.8\linewidth]{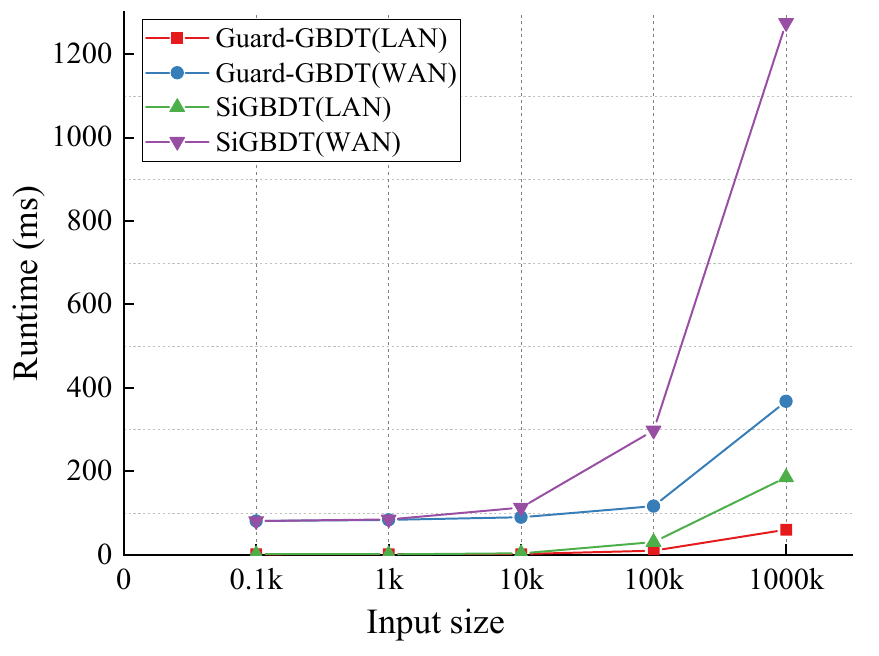}
    \caption{Gradient aggregation runtime}
    \label{fig:agg-comparison}
\end{figure}

\begin{figure*}[h]
  \centering
  \subfloat{\includegraphics[width=0.33\linewidth]{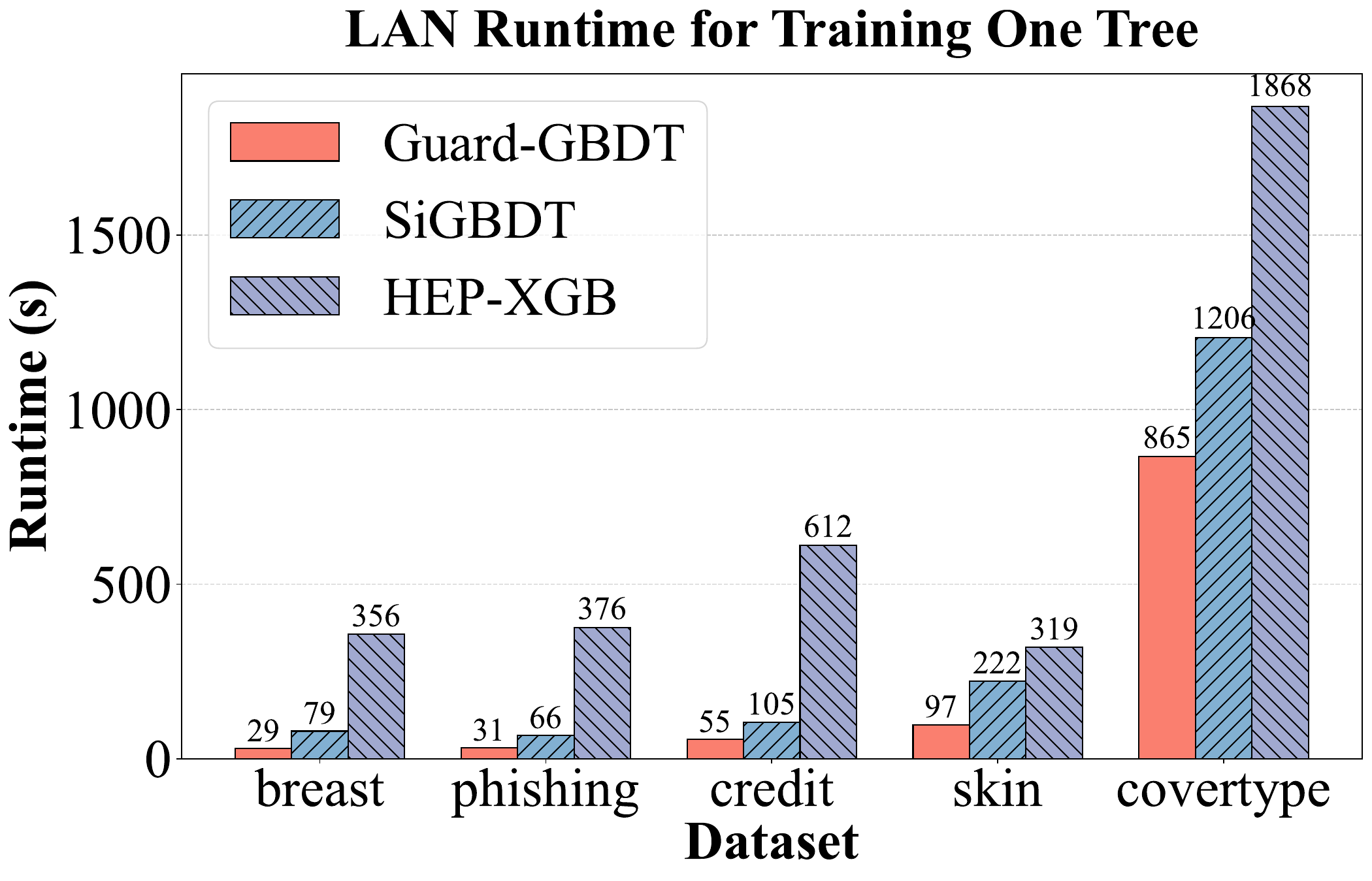}}
  \hfill
  \subfloat{\includegraphics[width=0.33\linewidth]{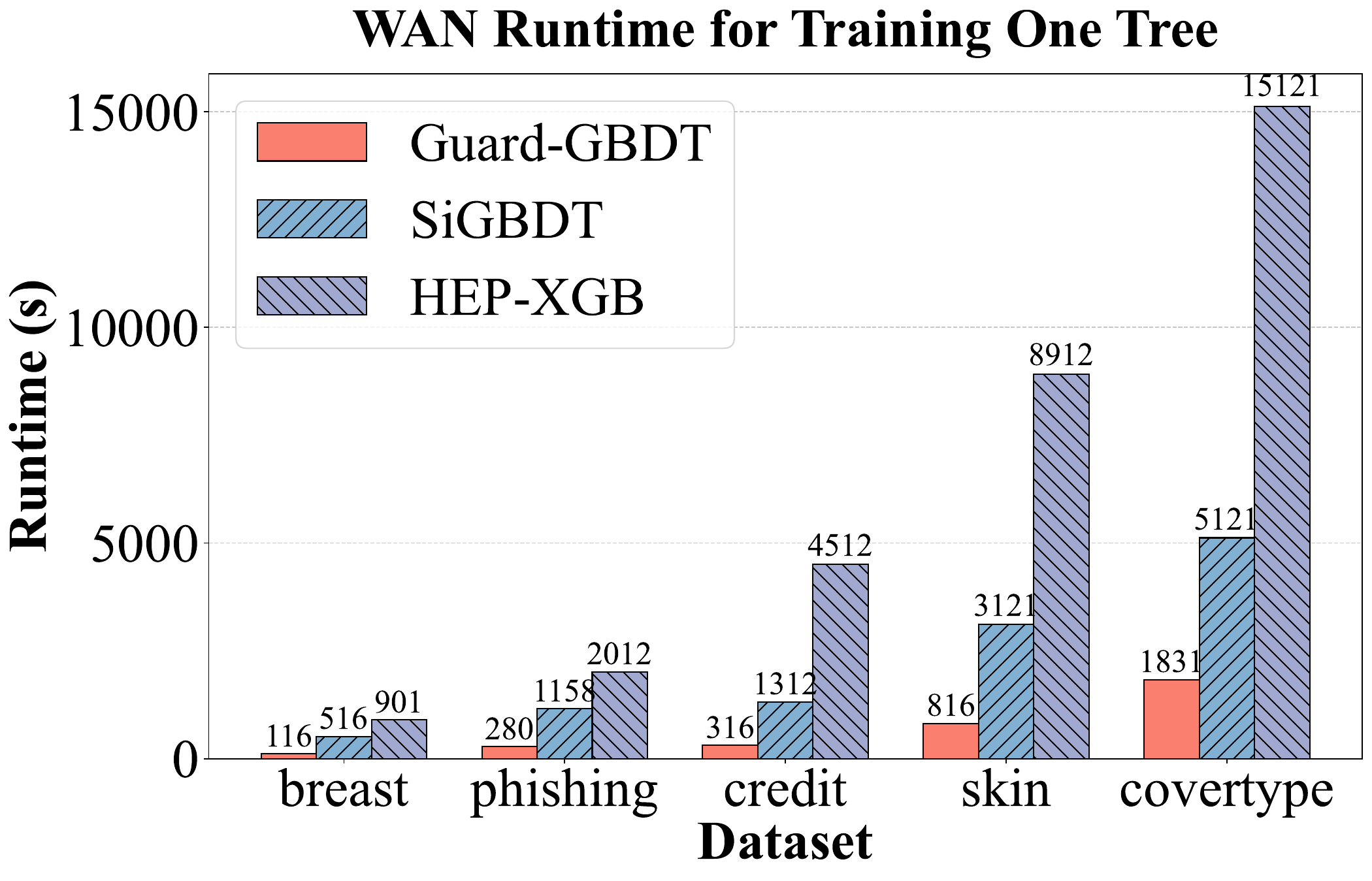}}
  \hfill
  \subfloat{\includegraphics[width=0.33\linewidth]{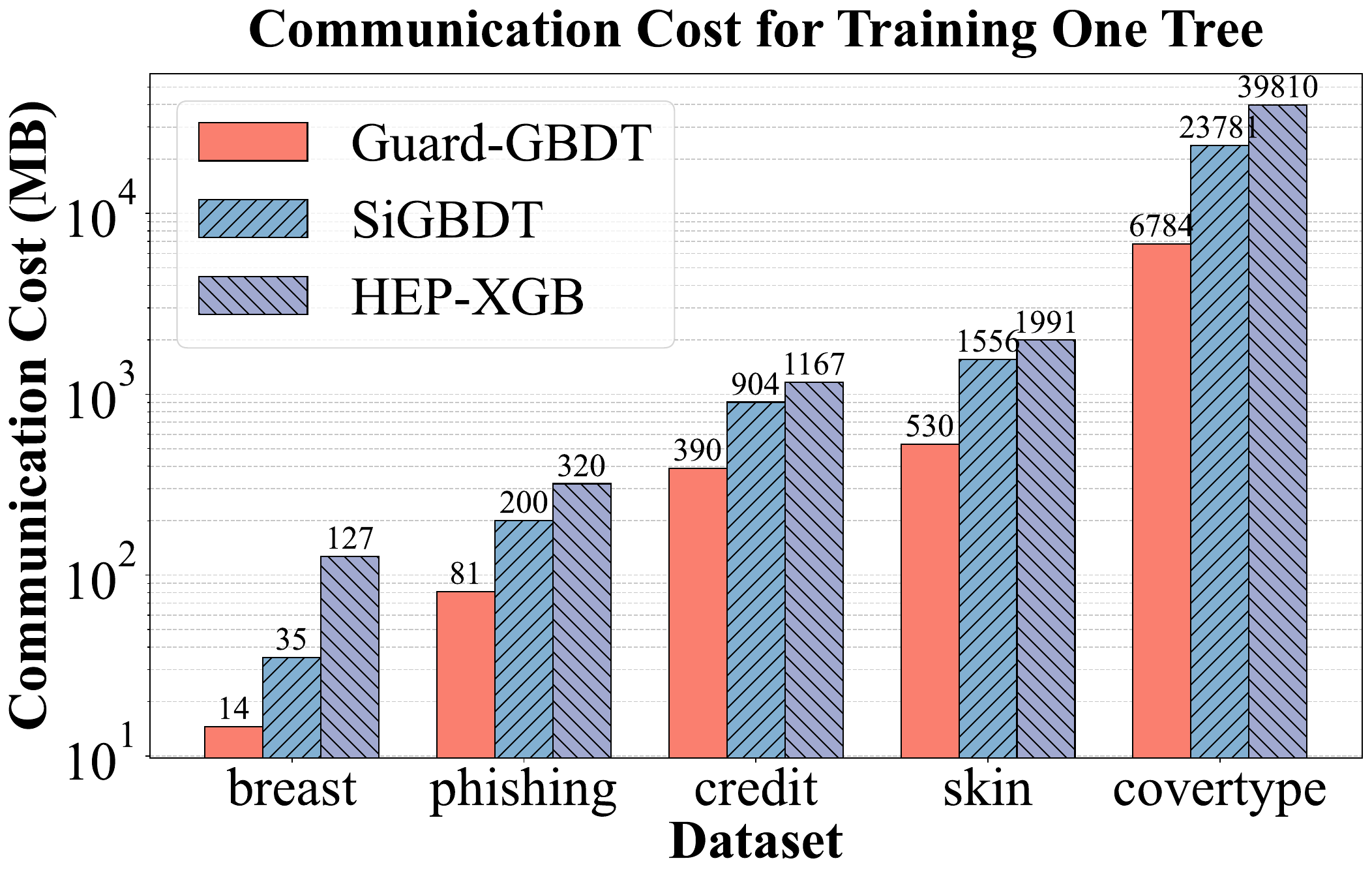}}
  \caption{Evaluation for training one tree}
  \label{fig:subfigures}
\end{figure*}

\begin{figure*}[h]
  \centering
  \subfloat{\includegraphics[width=0.32\linewidth]{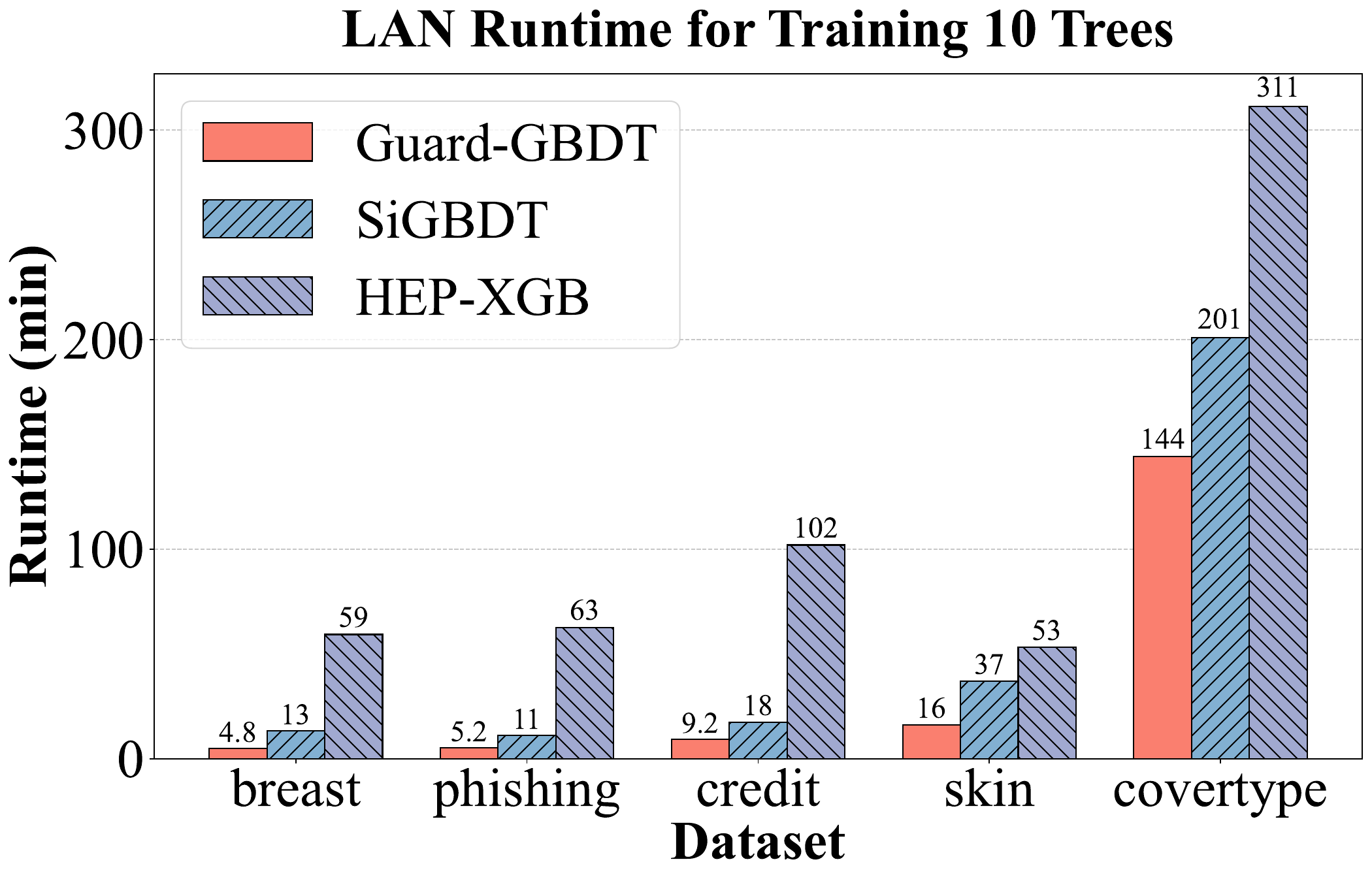}}
  \hfill
  \subfloat{\includegraphics[width=0.32\linewidth]{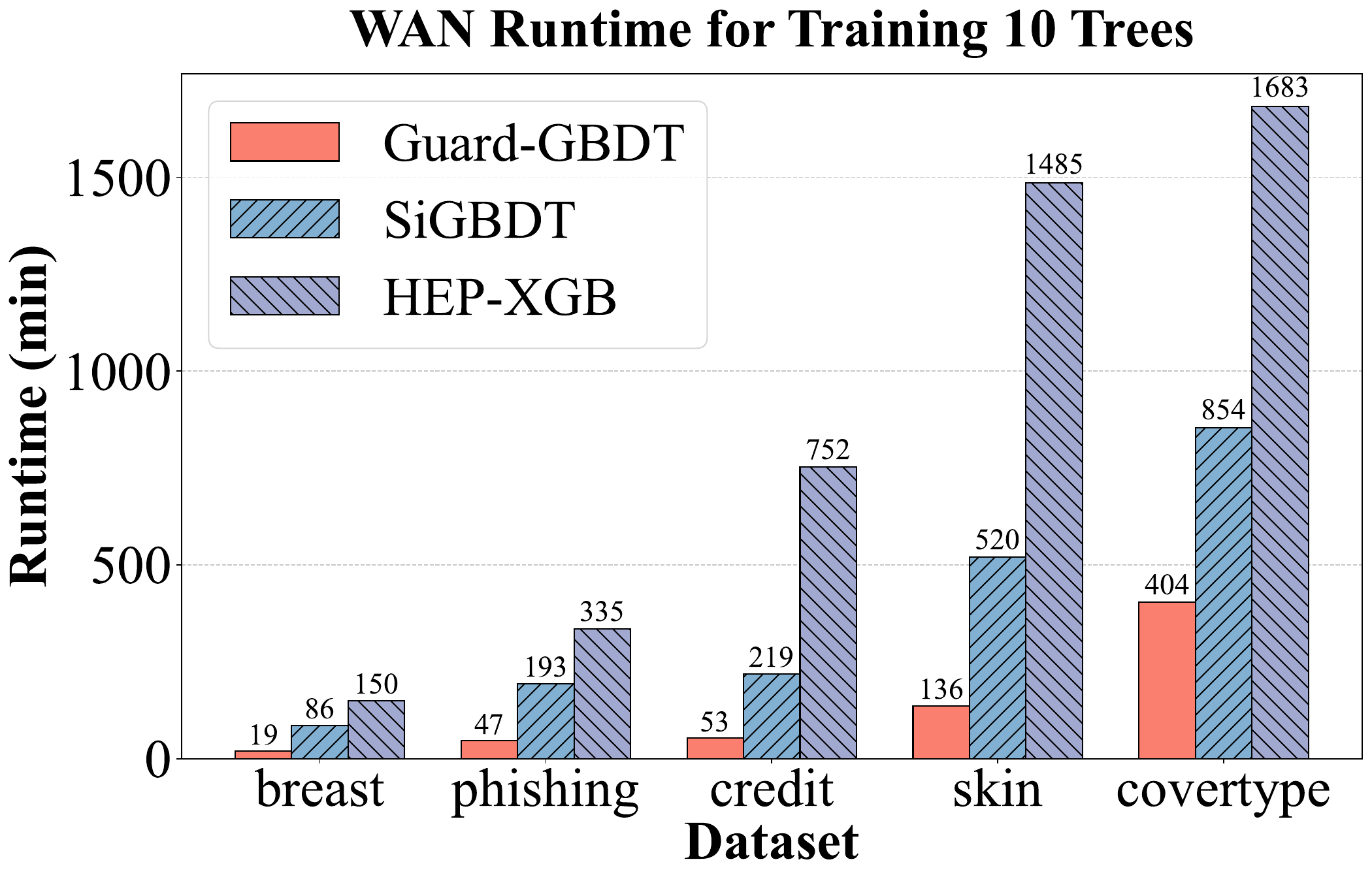}}
  \hfill
  \subfloat{\includegraphics[width=0.32\linewidth]{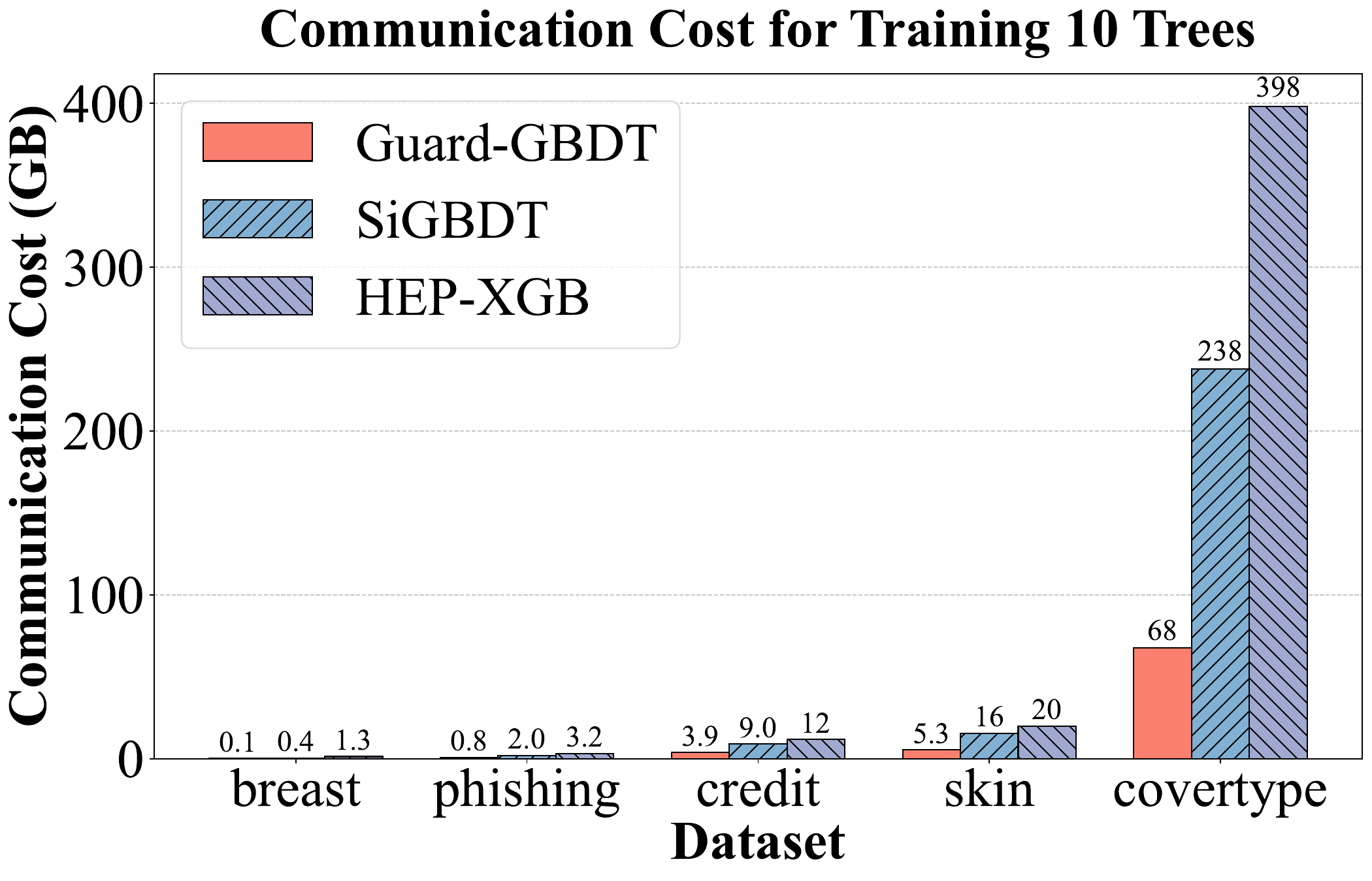}}
  \caption{Performance evaluation for training $T=10$ trees}
  \label{fig:10trees}
\end{figure*}

\begin{table}[h]
\scriptsize
\caption{Model accuracy over different numbers of trees}
\resizebox{1\linewidth}{!}{
\begin{tabular}{c c|ccc|c}
\toprule
\textbf{Number of trees} & \textbf{Dataset}       & \textbf{Guard-GBDT} & \textbf{SiGBDT} &\textbf{ HEP-XGB} & \textbf{XGBoost (Plain)} \\
\midrule
\multirow{4}{*}{T=5} &breast-cancer & \textbf{94.74}    & 93.86  & 87.72   & 96.37           \\
&credit        &   79.90  & \textbf{81.12}  &  76.34     & 82.42           \\
&phishing      &\textbf{ 89.42}    & 89.19  &   78.11    & 89.91           \\
&skin          &\textbf{100.00}      & 100.00    & 91,11    & 100.00       \\
&covertype     & \textbf{100.00}     & 100.00    & 90.81    & 100.00     \\
\hline
\multirow{4}{*}{T=10} &breast-cancer & \textbf{94.74}    & 94.26  & 90.72   & 97.37           \\
&credit        & 78.85    & \textbf{82.12}  & 77.14     & 82.37          \\
&phishing      & \textbf{90.28}    & 89.91  &  81.12    & 91.91           \\
&skin          & \textbf{100.00}     & 100.00    & 98,11    & 100.00       \\
&covertype     & \textbf{100.00}      & 100.00    & 99.81    & 100.00     \\
\hline
\multirow{4}{*}{T=15} &breast-cancer & \textbf{95.78}    & 94.86  & 89.72   & 97.37           \\
&credit        & 80.75    & \textbf{82.12}  &  78.14     & 82.15          \\
&phishing      &  \textbf{90.73}  & 89.91  &   89.12    & 92.11           \\
&skin          &\textbf{100.00}      & 100.00    &100.00    & 100.00       \\
&covertype     & \textbf{100.00}      & 100.00    & 100.00    & 100.00     \\
\hline
\multirow{4}{*}{T=20} &breast-cancer &\textbf{ 94.90 }   & 92.86  & 89.12   & 97.37           \\
&credit        & 80.80    & \textbf{81.12}  &  79.90     & 82.17           \\
&phishing      & \textbf{90.73}   & 89.91  &   89.12    & 91.91          \\
&skin          & \textbf{100.00}      & 100.00    & 100.00    & 100.00       \\
&covertype     & \textbf{100.00}      & 100.00    & 100.00    & 100.00     \\
\bottomrule
\end{tabular}}
\label{tab:accuacyt}
\end{table}

Next, we use the plaintext XGBoost to train a model over the $5$ real-world datasets as a baseline. Then, we compare our Guard-GBDT with SiGBDT and HEP-XGB under varying tree depths ($D = 4, 8, 16, 32$), where the segment size $n=12$ and tree number $T=5$. The hyperparameters and initialization are consistent for Guard-GBDT, SiGBDT, and HEP-XGB. 
Table~\ref{tab:accuracy}
shows that the accuracies of Guard-GBDT and SiGBDT are very similar under varying tree depths across multiple datasets, both are comparable to the accuracy of plaintext XGBoost, indicating strong model performance. However, HEP-XGB suffers significant accuracy degradation, underperforming Guard-GBDT by $6-12\%$ across datasets (e.g., $87.72\%$ vs. $94.74\%$ on breast cancer at $D=4$). The degradation stems from a numerically approximated sigmoid $sigmoid(x)\approx 0.5 + 0.5\cdot x/(1 + |x|)$ in HEP-XGB such that it has a larger error than Guard-GBDT and SiGBDT.

{Finally, we evaluate the accuracy of Guard-GBDT, SiGBDT, and HEP-GBDT under varying ensemble sizes of trees ($T=5, 10, 15, 20$) and compare them with the non-private XGBoost, as shown in 
Table~\ref{tab:accuacyt}.
HEP-XGB exhibits a strong dependence on $T$, with accuracy improvements of up to $12.89\%$ on the skin and $9.2\%$ on phishing as $T$ increased from $5$ to $20$.} However, Guard-GBDT and SiGBDT have minor fluctuations ($\pm1\%-\pm2\%$) and are comparable with the non-private XGBoost. They achieve perfect accuracy across the total number of trees. This indicates that Guard-GBDT and SiGBDT are stable and insensitive to the ensemble tree sizes.

\subsection{Microbenchmarks}
We evaluate component performance and compare it with SiGBDT and HEP-XGB in terms of runtime and communication costs in LAN and WAN, as shown in Fig.~\ref{fig:me}.
SiGBDT approximates \(\exp(x)\) using a Taylor series, followed by division for the sigmoid function, while HEP-XGB employs a numerical approximation \( sigmoid(x) \approx 0.5 + 0.5 \cdot x/(1 + |x|) \), both using Goldschmidt’s series for division. Their methods are less efficient than our lookup table-based approximation.

Fig.~\ref{fig:me} also shows that our approach significantly reduces communication overhead. Similarly, our lookup table-based leaf weight computation outperforms theirs in runtime and communication, as it relies only on scalar multiplication and lightweight FSS comparisons, avoiding SiGBDT and HEP-XGB’s costly division protocols. For splitting gain, we improve efficiency with a division-free method, outperforming SiGBDT and HEP-XGB. Guard-GBDT is \(788\times\) and \(1000 \times\) faster in LAN and \(814\times\) and \(905\times\) faster in WAN than SiGBDT and HEP-XGB, respectively.  For gradient aggregation, we optimize performance by compressing intermediate communication. In LAN, Guard-GBDT achieves \(2.0\times\) and \(65.8\times\) speedup over SiGBDT and HEP-XGB, while in WAN, it provides \(3.0\times\) and $164.4\times$ speedup, respectively.

Fig.~\ref{fig:me} also shows that the runtime for sigmoid, leaf weight, and split gain operations is not significantly different between LAN and WAN settings, indicating that computational cost, rather than communication overhead, is the primary bottleneck for these nonlinear MPC operations.  In contrast, the gradient aggregation runtime varies significantly between LAN and WAN settings. As shown in Fig.~\ref{fig:agg-comparison}, its runtime increases much more in WAN than in LAN, highlighting its sensitivity to communication delays. This confirms that communication compression in aggregation is both effective and necessary.

\subsection{Efficiency Evaluation}
As done in SiGBDT and HEP-XGB, to evaluate the efficiency, we test the training performance of Guard-GBDT over a single tree using $5$ real-world datasets and compare it with SiGBDT and HEP-XGB. As shown in Fig.~\ref{fig:subfigures},
Guard-GBDT is better than the other two works in terms of runtime and communication cost. The reason is that we leverage more streamlined lookup tables with $96B$ of memory\footnote{
When \(n \geq 12\) and \(\ell = 64\) in all environments in Fig.~\ref{fig:accuracy_under_segment_size}, Guard-GBDT achieves the best accuracy. Therefore, we set \(n = 12\) and \(\ell = 64\), with each lookup table requiring \(96B\) of memory to store approximations.}  to replace inefficient and complicated sigmoid function and division and use a communication-efficient protocol to reduce
data transmission by compressing intermediate messages during
gradient aggregation.

On the breast cancer dataset in the LAN network, Guard-GBDT is \(2.71 \times\) faster than SiGBDT and \(12.21 \times\) faster than HEP-XGB. For phishing, it is \(2.12 \times\) faster than SiGBDT and \(12.09 \times\) faster than HEP-XGB. On larger datasets like covertype, Guard-GBDT is \(2.11 \times\) and \(4.16 \times\) faster than SiGBDT and HEP-XGB, respectively. In WAN settings, Guard-GBDT is more efficient, running in 116s and outperforming HEP-XGB and SiGBDT on the breast cancer dataset. For covertype, Guard-GBDT is \(2.7 \times\) and \(8.2 \times\) faster than the other two methods. In terms of communication, Guard-GBDT requires 14 MB for breast cancer, saving \(2.5 \times\) compared to SiGBDT and \(9.07 \times\) compared to HEP-XGB. For covertype, Guard-GBDT reduces communication overhead by \(3.5 \times\) and \(5.86 \times\) compared to SiGBDT and HEP-XGB, respectively.

To evaluate performance of training multiple trees, we tested our approach by measuring the runtime and communication bandwidth for training 10 trees, as shown in Fig.~\ref{fig:10trees}. The results shows Guard-GBDT also is best one compared with the other two works. In the LAN network, Guard-GBDT outperforms SiGBDT and HEP-XGB across all datasets. For the breast cancer dataset, it is $2.75\times$ faster than SiGBDT and $12.36\times$ faster than HEP-XGB. For phishing, Guard-GBDT is $2.12\times$ faster than SiGBDT and $12.05\times$ faster than HEP-XGB, while for the larger covertype dataset, it is $1.39\times$ faster than SiGBDT and $2.16\times$ faster than HEP-XGB. In the WAN network, Guard-GBDT is also more efficient, running $4.46\times$ faster than SiGBDT and $7.79\times$ faster than HEP-XGB, and for covertype, it is \(2.11 \times\) faster than SiGBDT and \(4.16 \times\) faster than HEP-XGB. In terms of communication costs, Guard-GBDT requires 0.1 GB for breast cancer, saving \(4 \times\) the overhead of SiGBDT and \(13 \times\) that of HEP-XGB. For phishing, it saves \(2.5 \times\) over SiGBDT and \(4 \times\) over HEP-XGB, and for covertype, Guard-GBDT reduces communication costs by \(3.51 \times\) than SiGBDT and \(5.87 \times\) than HEP-XGB.

\begin{table*}
\footnotesize
\caption{Evaluation for training one tree with synthetic datasets under scalable settings.}
\label{tab:endtoendcomparison}
\setlength{\tabcolsep}{4.7pt}
\begin{tabular}{ccccccccccccccc}
\toprule
\multicolumn{1}{c}{\multirow{2}{*}{N}} & \multicolumn{1}{c}{\multirow{2}{*}{F}} & \multicolumn{1}{c}{\multirow{2}{*}{B}} & \multicolumn{1}{c}{\multirow{2}{*}{D}} & \multicolumn{3}{c}{Runtime on LAN (s) } & \multirow{2}{*}{} & \multicolumn{3}{c}{Runtime on WAN (s) } & \multirow{2}{*}{} & \multicolumn{3}{c}{Communaction cost (MB)} \\ \cline{5-7} \cline{9-11} \cline{13-15}
\multicolumn{1}{c}{}                   & \multicolumn{1}{c}{}                   & \multicolumn{1}{c}{}                   & \multicolumn{1}{c}{}                   & \textbf{Guard-GBDT}   & \textbf{SiGBDT}  & \textbf{HEP-XGB}  &                   & \textbf{Guard-GBDT}   &\textbf{ SiGBDT}  & \textbf{HEP-XGB}  &                   &\textbf{Guard-GBDT}     & \textbf{SiGBDT}    & \textbf{HEP-XGB}    \\ \midrule
10k                                    & 10                                     & 8                                      & 4                                      &    20.41         &   38.46        &     286.18     &                   &   288.93   &  573.08    &   964.52  &                   &  21.51      & 95.42          &   172.94         \\
50k                                    & 10                                     & 8                                      & 4                                      &     34.97        &    77.68     &  326.83       &                   &       392.55     &    699.72     &   1107.70       &                   &  105.93         & 475.64        &       768.99     \\
10k                                    & 20                                   & 8                                      & 4                                      &      30.21            &  64.56        &  530.07        &                   & 433.08   &     932.75     &   1794.44          &              &   42.26    &      189.22     &         319.78    \\
10k                                    & 10                                     & 16                                     & 4                                      &   25.10          &   51.16       &     369.22       &                   &     315.94        &   551.36      &         1160.52    &                & 42.35            &     189.56      &     310.64       \\

10k                                    & 10                                     & 8                                      & 5                                      &  32.49         &  68.12       &   598.46       &                   &    580.19     &  1159.55   &  1,929.04      &                   &    46.48         & 204.29  &   346.07          \\
\bottomrule
\end{tabular}
\end{table*}

\subsection{Scalability Evaluation}
{As done in SiGBDT,  we generate random synthetic datasets to evaluate scalability, as real-world data often fails to simultaneously satisfy specific constraints on dataset size ($N$) and feature size ($F$), such as $N = 50$k and $F = 10$.
The default experimental parameters are configured as follows: a dataset size of $N = 10$k, a tree depth of $D = 4$, a feature size of $F = 10$, and a bucket size of $B = 8$.
In scalability evaluation, we vary training parameters based on the default settings.}
For clarity, the table \ref{tab:endtoendcomparison} presents a comprehensive comparison of scalability evaluation over one tree for Guard-GBDT, SiGBDT, and HEP-XGB. The results show that our Guard-GBDT outperforms other two works in terms of runtime and communication cost.

For the baseline,  Guard-GBDT achieves $1.88 \times$ and $14.02 \times$ improvement over SiGBDT and HEP-XGB, respectively.
With sample size $N$, feature size $F$, bucket size $B$, and tree depth $D$ increasing, the performance of Guard-GBDT is more significant. When the dataset size is $N=50$k with the same configuration, Guard-GBDT is $2.22 \times$ faster than SiGBDT and $9.34 \times$ faster than HEP-XGB. When training a tree of depth $D=5$, Guard-GBDT is $2.3 \times$ and $18.7 \times$ faster than SiGBDT and  HEP-XGB, respectively. In terms of communication,  the baseline test of Guard-GBDT’s communication cost improves $4.43 \times$ and $8.04 \times$ than SiGBDT and HEP-XGB, respectively. As different configurations increase, Guard-GBDT’s communication cost is lower. When training a tree of depth $D=5$ on $10$k samples with $F=10$ features and $B=8$ buckets, Guard-GBDT is $4.48 \times$ and $7.57 \times$  faster than SiGBDT and HEP-XGB, respectively. 

%% file: 7-conclusion.tex
\section{Conclusion}
This paper presents Gaurd-GBDT, a high-performance privacy-preserving framework for GBDT training. To improve efficiency, we design a new MPC-friendly approach that removes inefficient divisions and sigmoid functions in MPC. By these techniques, we accelerate the training process without compromising the privacy guarantees inherent to MPC. Besides, we present a communication-friendly aggregation protocol that can compress intermediate communication messages during gradient aggregation. It addresses the explosive growth of communication overhead as the data scale increases. Extensive experiments demonstrate that Guard-GBDT surpasses the state-of-the-art HEP-MPC and SiGBDT frameworks in both LAN and WAN settings. In terms of accuracy, Guard-GBDT achieves results comparable to SiGBDT and plaintext XGBoost, while outperforming HEP-XGB.

\section*{Acknowledgments}
This paper is supported by the National Natural Science Foundation of China (No. 62402358, No. 62220106004), the Technology Innovation Leading Program of Shaanxi (No. 2023KXJ-033), the Key R\&D Program of Shandong Province of China (No. 2023CXPT056), the Young Talent Fund of Association for Science and Technology in Shaanxi, China (No. 20240138), the Open Topics from the Lion Rock Labs of Cyberspace Security (under the project \#LRL24004), the Open Project Foundation of Shaanxi Key Laboratory of Information Communication Network and Security (No. ICNS202301), the Fundamental Research Funds for the Central Universities (No. ZDRC2202, ZYTS25081, KYFZ25005), the Xidian University Specially Funded Project for Interdisciplinary Exploration (No. TZJHF202502) and the Double First-Class Overseas Research Project of Xidian University.

%% file: Appendx.tex
\appendices

\section{Correctness and Security of FSS}
\label{app:cs-fss}
\begin{definition}[Correctness and Security of FSS\cite{boyle2016function}] Let $f(x)$ be one description in the function family $\mathcal{F}$, i.e, $f(x)\in \mathcal{F}$.
We say that (\textbf{FSS.Gen, FSS.Eval}) as in Definition \ref{fSS-Sysntax} is an
FSS scheme for $\mathcal{F}$ with respect to leakage (\(Leak\)) if it satisfies the following.
\begin{itemize}
    \item \textbf{Correctness}: For all $\hat{f}(x) $ describing $f(x)$: $\mathbb{G}^{in}\rightarrow \mathbb{G}^{out}$, and every $x\in \mathbb{G}^{in}$, if $(\mathcal{K}_0,\mathcal{K}_1)\gets\text{FSS.Gen}(1^{\lambda},\hat{f}(x))$, then
    Pr$[\text{FSS.Eval}(0,\mathcal{K}_0,x)+\text{FSS.Eval}(1,\mathcal{K}_1,x)=f(x)]=1$,
    where \( \mathbb{G}^{\text{in}} \) and \( \mathbb{G}^{\text{out}} \) represent the input and output domains, respectively, of the function \( \hat{f}(x)\).
    \item \textbf{Security:} 
There exists a probabilistic polynomial-time (PPT) algorithm simulator \( \text{Sim}_b \) for $b\in\{0,1\}$. This simulator ensures that for any sequence \( \{f_\phi(x)\}_{\phi \in \mathbb{N}} \) of polynomial-size function descriptions from \( \mathcal{F} \) and corresponding polynomial-size input sequences \( x_\phi \) for \( f_\phi(x)\), the outputs from the following \textbf{Real} and \textbf{Ideal} experiments are computationally indistinguishable:
    
--\textbf{Real:}\( (\mathcal{K}_0, \mathcal{K}_1) \gets \text{Gen}(1^\lambda, \hat{f}_\phi(x)) \) and output \( \mathcal{K}_b \).

--\textbf{Ideal:} Output \( \text{Sim}_b(1^\lambda, Leak(\hat{f}_\lambda(x)))\).

\end{itemize}
\end{definition} 

\section{Security Analysis}
\label{appendix:security}
Guard-GBDT follows the standard definition of security against semi-honest adversaries (Definition \ref{PPT-definition}). We can prove the security of our protocols using the real-world/idea-word simulation paradigm\cite{goldreich2019play} in a bottom-up fashion:
\begin{enumerate}
  \item First, we instantiate the two cryptographic primitives of ASS and FSS that we used, including secure multiplication, secure addition, and DCF. For these primitives, 
  we can refer to ABY\cite{demmler2015aby} to derive the security of ASS and Orca\cite{jawalkar2024orca} to derive the security of FSS under the semi-honest model
  \item Next, we prove the security of our proposed protocols in Theorem \ref{th:sigmoid}-\ref{th:agg}.
  \item Finally, we put everything together and prove the security of the secure training protocol of Guard-GBDT in Theorem \ref{th:training}.
\end{enumerate}

To prove the security of all proposed protocols under the semi-honest adversary model, we employ the sequential composition theory \cite{canetti2000security}, which is a tool that enables us to analyze the security of
cryptographic protocols in a modular approach. Let a protocol $\prod_{\mathcal{F}}$ be securely computed by invoking sub-protocols $\prod_{f_1},\ldots,\prod_{f_m}$ for computing $f_i,\ldots,f_m$. The theorem states
that it suffices to consider the execution of $\prod_{\mathcal{F}}$ in a ($f_1,\ldots,f_m$)-hybrid model. If all invoked sub-protocols $\prod_{f_1},\ldots,\prod_{f_m}$ is secure against the semi-honest adversaries in the hybrid model, then $\prod_{\mathcal{F}}$ protocol also secure against the semi-honest adversaries.  Next, we provide the security proofs for our cryptographic protocols.

\begin{theorem}
\label{th:sigmoid}
If the arithmetic operations of ASS and DCF of FSS are secure against semi-honest adversaries, then our $\Pi_{\text{LUT}_{\delta}^{n}}$ protocol and $\Pi_{\text{LUT}_{w}^{n}}$ protocol are secure under the semi-honest adversaries model. 
\end{theorem}

\begin{proof}
We first prove the security of $\Pi_{\text{LUT}_{\delta}^{n}}$ protocol. In preprocessing on the offline phase, the secure third party (STP) provides only the randomness used to generate keys for the online computation and does not receive any private information from the computing parties $P_0$ and $P_1$. Hence, the offline phase of this protocol is trivially secure against the semi-honest corruption of STP. To prove the security of this protocol in the online phase,  we construct simulators in two distinct cases depending on which party is
corrupted. For all PPT adversaries, the corrupted party’s view from the interaction between $P_0$ and $P_1$ is indistinguishable from its view when interacting with a simulator instead. Deriving from $\Pi_{\text{LUT}_{\delta}^{n}}$ protocol, we see that
the interactive messages between $P_0$ and $P_1$ include
$m_i = x+\alpha_i-\omega_i$ (Line \ref{alg:sec_sigmoid_line_3} in Algorithm \ref{alg:secure-sigmoid-protocol}) for $i\in\{0,n-1\}$ and data transferred in executing the arithmetic operations and DCF. 
Since the existing works have proved that the arithmetic operations and DCF are secure against semi-honest adversaries, any PPT adversary cannot distinguish the simulator’s views from the data transferred in these operations. As for the interactive messages $\{m_i\}$, we can construct a simulator $\mathcal{S}_0$ to simulate $P_0$'s view. $S_0$ chooses $n$ random integers $\{r'_{i}\in\mathbb{Z}_{2^{\ell}}\}$ to simulate $n$ interactive messages $\{m_i = x+\alpha_i-\omega_i\}$. Recall that the matrices $\{\alpha_i\}$ are generated randomly in the offline phase. Thus $\{r'_i\}$ and  $\{m_i = x+\alpha_i-\omega_i\}$ are indistinguishable such that any PPT adversary cannot distinguish from its view when interacting with the simulator $\mathcal{S}_0$ instead.  Clearly, a simulator $\mathcal{S}_1$ can do the same work as above to simulate $P_1$'s view. According to the composition theory, we can claim that $\Pi_{\text{LUT}_{\delta}^{n}}$ protocol is secure under the semi-honest adversaries model. 

Since the implementation and execution steps of both protocols are similar, the security proof for the  $\Pi_{\text{LUT}_{w}^{n}}$protocol is analogous to that of the $\Pi_{\text{LUT}_{\delta}^{n}}$ protocol. We can therefore claim that the  $\Pi_{\text{LUT}_{w}^{n}}$ protocol is also secure under the semi-honest adversaries model.
\end{proof}

 \begin{theorem}
 \label{th:agg}
If the arithmetic operations of ASS are secure against semi-honest adversaries, then our $\Pi_{\text{Agg}}$ protocol is secure under the semi-honest adversaries model.
\end{theorem}

\begin{proof}
Similar to Theorem \ref{th:sigmoid}, we will prove the security of $\Pi_{\text{Agg}}$ protocol briefly.  In preprocessing on the offline phase, STP provides only randomness used to generate keys and does not receive any private information from the $P_0$ and $P_1$. Hence, the offline phase of this protocol is trivially secure against the semi-honest corruption of STP. In the online phase,
the interactive messages between $P_0$ and $P_1$ include data transferred in executing the arithmetic operations and $\{\hat{s}_i=s_i\oplus r_{s_i}, \hat{g}'_i= g_i+r_i+2^{\ell'} \mod 2^{\ell'} \}$ (Line \ref{alg:agg_line_1} and \ref{alg:agg_line_2} in Algorithm \ref{alg:agg}) for $i\in [0,N]$. Since the arithmetic operations of ASS are secure against semi-honest adversaries, any PPT adversary cannot distinguish the simulator’s views from the data transferred in these operations.
As for the interactive messages $\{\hat{s}_i=s_i\oplus r_{s_i}, \hat{g}'_i= g_i+r_i+2^{\ell'} \mod 2^{\ell'} \}$ in which $r'_{s_i}$  and $r_i$ are randomness form offline phase,   we can construct a simulator $\mathcal{S}_b$ to simulate each $P_b$'s view, where $b\in \{0,1\}$. The simulator $S_b$ chooses $N$ pairs of randomness  $r'_{s_i}\in\{0,1\}$ and $r'_{g_i} \mod 2^{\ell'}$ to simulate $\{\hat{s}_i, \hat{g}'_i\}$. Since the randomness  $r_i$  and $r_{s_i}$ are random such that $r'_{g_i}$ and  $\hat{g}'_i$ are indistinguishable  and  $r'_{s_i}$ and  $\hat{s}_i$ are also indistinguishable. Thus, any PPT adversary cannot distinguish from its view when interacting with the simulator $\mathcal{S}_b$ instead. According to the composition theory, we can claim that $\Pi_{\text{Agg}}$ protocol is secure under the semi-honest adversaries model 
\end{proof}

 \begin{theorem}
\label{th:training}
As long as our proposed sub-protocols are secure against semi-honest adversaries, Guard-GBDT is secure under the semi-honest adversaries model.
\end{theorem}
\begin{proof}
Guard-GBDT are stitched together from our proposed sub-protocols.  Beyond the input and output of the protocols, the interactive messages of Guard-GBDT consist only of intermediate values and outputs for subprotocols. Since the security of these sub-protocols is proved in Theorem \ref{th:sigmoid}-\ref{th:agg}, their intermediate messages are individual and uniformly random.  In the case of corrupted $P_0$  or $P_1$, the simulator $\mathcal{S}_b$ can choose uniformly random values to simulate these intermediate messages such that any PPT adversary cannot distinguish from its view when interacting with the simulator $\mathcal{S}_b$ instead.  According to the composition theory, we can claim that the Guard-GBDT is secure under semi-honest adversary model.
\end{proof}

\section{Correctness of $\mathcal{G}'$ in Guard-GBDT}
\label{app:correctness-g}
In the general GBDT, the second-order approximation can be used to optimize the objective function $\mathcal{L}$  quickly in the training task:

\begin{equation}
\begin{aligned}
    \mathcal{L}^{(t)} & \approx  
    &= \sum_{j=1}^{K} \left[ \left( \sum_{i \in X} g_i \right) w_j + \frac{1}{2} \left( \sum_{i \in X} h_i +\gamma \right) w_j^2 \right]
\end{aligned}
\end{equation}
where \(K\) denotes the total number of leaves in the tree, \(w_j\) is the leaf weight at the \(j\)-th leaf node, \( X \) is the set of samples in leaf node \( j \), and \(\gamma\) is a constant.
Assuming the tree structure is fixed, the objective function can be simplified into a quadratic function with respect to the leaf weights \( w_j \). For each leaf node \( j \), the optimal weight is:
\begin{equation}
\label{eq-wj}
    w_j^* = -\frac{\sum_{i \in X} g_i}{\sum_{i \in X} h_i + \lambda} = -\frac{G_X}{H_X+\lambda}
\end{equation}
Substituting the weights, the minimum of $\mathcal{L}^{(t)}$ at single one leaf node (i.e., $K=1$) is:
\begin{equation}
    \mathcal{L}^{(t)}= = -\frac{1}{2}\cdot \frac{(\sum_{i \in X} g_i)^2}{\sum_{i \in X} h_i + \gamma} =  -\frac{1}{2}\cdot \frac{(G_{X})^2}{H_{X}+\gamma}
\end{equation}
Assume that \( L \) and \( R \) are the instance sets of the left and right nodes after the split. Let \( X = L \cup R \). Before splitting, the minimum of $\mathcal{L}^{(t)}$ is as follows:
\begin{equation}
\label{eq:L-X}
    \mathcal{L}^{*}_X = -\frac{(G_{X})^2}{2(H_{X}+\gamma)} .
\end{equation}
After splitting,  the minimum of $\mathcal{L}^{(t)}$ is as follows:
\begin{equation}
\label{eq:L-RL}
\begin{aligned}
    &\mathcal{L}^{*}_L= -\frac{(G_{L})^2}{2(H_{L}+\gamma)}; \mathcal{L}^{*}_R= -\frac{(G_{R})^2}{2(H_{R}+\gamma)}.
\end{aligned}
\end{equation}
Based on Eq.~\ref{eq:L-X} and \ref{eq:L-RL}, the gain score $\mathcal{G}$ can be computed as:
\begin{equation}
    \begin{aligned}
    \mathcal{G}&=\mathcal{L}^{*}_X-(\mathcal{L}^{*}_L+\mathcal{L}^{*}_R)\\
    &= \frac{1}{2}(\frac{({G_L})^2}{H_{L}+\gamma}+\frac{(G_{R})^2}{H_{R}+\gamma}-\frac{(G_{X})^2}{H_{X}+\gamma}).
    \end{aligned}
\end{equation}

When the parent node is fixed, \( G_X \) and \( H_X \) are constants such that $\mathcal{L}^*_X=-\frac{(G_{X})^2}{H_{X}+\gamma}$ also is a constant since the dataset determines them and does not change with the split. 
Thus, we have:
\begin{equation}
    \begin{aligned}
    \mathcal{G}&\equiv-(\mathcal{L}^{*}_L+\mathcal{L}^{*}_R)= \frac{({G_L})^2}{H_{L}+\gamma}+\frac{(G_{R})^2}{H_{R}+\gamma}\\
    &=[(H_R+\lambda)({G_L})^2+(H_L+\lambda)({G_L})^2]\cdot \frac{1}{(H_{L}+\gamma)(H_{R}+\gamma)}\\
    \end{aligned}
\end{equation}
 Given $H_X=H_L+H_R$ at the fixed parent node, the derivative of $f(H_L,H_R)={1}/{((H_{L}+\gamma)(H_{R}+\gamma))}$ is:
\begin{equation}
\label{eq:derivative}
    \frac{\partial f(H_L,H_R)}{\partial H_L}=\frac{H_X-2H_L}{((H_{L}+\gamma)(H_{R}+\gamma))^2}
\end{equation}
Following on Eq.~\ref{eq:derivative}, the function  $f(H_L, H_R)$ is monotonically increasing in \( H_L \) when \( 2H_L < {H_X} \) and monotonically decreasing when \( 2H_L \geq {H_X} \). This implies that it is strictly monotonic in each segmented range. Based on this,  we have:
\begin{equation}
\mathcal{G}\propto\mathcal{G}'=
\label{eq:our-gain}
   \begin{cases}
        G^{*} & 2H_L<H_X\\
        -G^{*} &2H_L\geq H_X
   \end{cases}
\end{equation}
where $\mathcal{G}^*=(H_R+\lambda)({G_L})^2+(H_L+\lambda)({G_L})^2$ and $\propto$ denotes that $\mathcal{G}$ is proportional to $\mathcal{G}'$, meaning that the ordering among candidates remains unchanged. Thus, $\mathcal{G}'$ is equivalent to $\mathcal{G}$.

\section{Accuracy with Different Segments on $T=10$}
\label{app:acc-sig}
We test model accuracy with different segment sizes on five real datasets to find the best segment size of approximations, where the total number of trees in GBDT is $T=10$. As shown in Fig.~\ref{fig:accuracy_under_segment_size_10}, the results for large dataset demonstrate still stable performance with accuracy close to 100\% across various segment sizes ($n$) and tree depths ($D$), indicating that these datasets are insensitive to changes in segment size. In contrast, results on smaller datasets exhibit fluctuations and are more sensitive to segment size changes. As the segment size grows, the accuracy tends to stabilize, suggesting that once the segment size exceeds a threshold (i.e., $ n \geq10$ for $D=4$ and $ n \geq10$ for $D=8,16~\text{or}~32$ ), the model's accuracy becomes consistent across different datasets.

%% file: paper.bib
@article{guo2023gfs,
  title={GFS-CNN: A GPU-friendly secure computation platform for convolutional neural networks},
  author={Guo, Chao and Cheng, Ke and Fu, Jiaxuan and Fan, Ruolu and Chang, Zhao and Zhang, Zhiwei and Song, Anxiao},
  journal={Journal of Networking and Network Applications},
  volume={3},
  number={2},
  pages={66--72},
  year={2023},
  publisher={Institute of Electronics and Computer}
}

@article{song2024secnet,
  title={L-secnet: Towards secure and lightweight deep neural network inference},
  author={Song, Anxiao and Fu, Jiaxuan and Mu, Xutong and Zhu, XingHui and Cheng, Ke},
  journal={Journal of Networking and Network Applications},
  volume={3},
  number={4},
  pages={171--181},
  year={2024},
  publisher={Institute of Electronics and Computer}
}

@inproceedings{boyle2021function,
  title={Function secret sharing for mixed-mode and fixed-point secure computation},
  author={Boyle, Elette and Chandran, Nishanth and Gilboa, Niv and Gupta, Divya and Ishai, Yuval and Kumar, Nishant and Rathee, Mayank},
  booktitle={Annual International Conference on the Theory and Applications of Cryptographic Techniques},
  pages={871--900},
  year={2021},
  organization={Springer}
}

@inproceedings{fereidooni2021safelearn,
  title={SAFELearn: Secure aggregation for private federated learning},
  author={Fereidooni, Hossein and Marchal, Samuel and Miettinen, Markus and Mirhoseini, Azalia and M{\"o}llering, Helen and Nguyen, Thien Duc and Rieger, Phillip and Sadeghi, Ahmad-Reza and Schneider, Thomas and Yalame, Hossein and others},
  booktitle={2021 IEEE Security and Privacy Workshops (SPW)},
  pages={56--62},
  year={2021},
  organization={IEEE}
}

@INPROCEEDINGS{9155279,
  author={Cheng, Ke and Wang, Liangmin and Shen, Yulong and Liu, Yangyang and Wang, Yongzhi and Zheng, Lele},
  booktitle={IEEE INFOCOM 2020 - IEEE Conference on Computer Communications},
  title={A Lightweight Auction Framework for Spectrum Allocation with Strong Security Guarantees},
  year={2020},
  volume={},
  number={},
  pages={1708-1717},
  doi={10.1109/INFOCOM41043.2020.9155279}}

@inproceedings{catrina2010secure,
  title={Secure computation with fixed-point numbers},
  author={Catrina, Octavian and Saxena, Amitabh},
  booktitle={Financial Cryptography and Data Security: 14th International Conference, FC 2010, Tenerife, Canary Islands, January 25-28, 2010, Revised Selected Papers 14},
  pages={35--50},
  year={2010},
  organization={Springer}
}

@article{feng2018privacy,
  title={Privacy-preserving tensor decomposition over encrypted data in a federated cloud environment},
  author={Feng, Jun and Yang, Laurence T and Zhu, Qing and Choo, Kim-Kwang Raymond},
  journal={IEEE Transactions on Dependable and Secure Computing},
  volume={17},
  number={4},
  pages={857--868},
  year={2018},
  publisher={IEEE}
}

@inproceedings{lu2023squirrel,
  title={Squirrel: A Scalable Secure $\{$Two-Party$\}$ Computation Framework for Training Gradient Boosting Decision Tree},
  author={Lu, Wen-jie and Huang, Zhicong and Zhang, Qizhi and Wang, Yuchen and Hong, Cheng},
  booktitle={32nd USENIX Security Symposium (USENIX Security 23)},
  pages={6435--6451},
  year={2023}
}

@article{zheng2023privet,
  title={Privet: A privacy-preserving vertical federated learning service for gradient boosted decision tables},
  author={Zheng, Yifeng and Xu, Shuangqing and Wang, Songlei and Gao, Yansong and Hua, Zhongyun},
  journal={IEEE Transactions on Services Computing},
  volume={16},
  number={5},
  pages={3604--3620},
  year={2023},
  publisher={IEEE}
}

@article{xu2024elxgb,
  title={ELXGB: An Efficient and Privacy-Preserving XGBoost for Vertical Federated Learning},
  author={Xu, Wei and Zhu, Hui and Zheng, Yandong and Wang, Fengwei and Zhao, Jiaqi and Liu, Zhe and Li, Hui},
  journal={IEEE Transactions on Services Computing},
  year={2024},
  publisher={IEEE}
}

@inproceedings{fang2021large,
  title={Large-scale secure XGB for vertical federated learning},
  author={Fang, Wenjing and Zhao, Derun and Tan, Jin and Chen, Chaochao and Yu, Chaofan and Wang, Li and Wang, Lei and Zhou, Jun and Zhang, Benyu},
  booktitle={Proceedings of the 30th ACM International Conference on Information \& Knowledge Management},
  pages={443--452},
  year={2021}
}

@inproceedings{jiang2024sigbdt,
  title={SiGBDT: Large-Scale Gradient Boosting Decision Tree Training via Function Secret Sharing},
  author={Jiang, Yufan and Mei, Fei and Dai, Tianxiang and Li, Yong},
  booktitle={Proceedings of the 19th ACM Asia Conference on Computer and Communications Security},
  pages={274--288},
  year={2024}
}

@article{tian2023sf,
  title={FederBoost: Private Federated Learning for GBDT},
  author={Tian, Zhihua and Zhang, Rui and Hou, Xiaoyang and Lyu, Lingjuan and Zhang, Tianyi and Liu, Jian and Ren, Kui},
  journal={IEEE Transactions on Dependable and Secure Computing},
  year={2023},
  publisher={IEEE}
}

@article{cheng2021secureboost,
  title={Secureboost: A lossless federated learning framework},
  author={Cheng, Kewei and Fan, Tao and Jin, Yilun and Liu, Yang and Chen, Tianjian and Papadopoulos, Dimitrios and Yang, Qiang},
  journal={IEEE intelligent systems},
  volume={36},
  number={6},
  pages={87--98},
  year={2021},
  publisher={IEEE}
}

@inproceedings{lindell2000privacy,
  title={Privacy preserving data mining},
  author={Lindell, Yehuda and Pinkas, Benny},
  booktitle={Annual international cryptology conference},
  pages={36--54},
  year={2000},
  organization={Springer}
}

@article{abspoel2021secure,
  title={Secure training of decision trees with continuous attributes},
  author={Abspoel, Mark and Escudero, Daniel and Volgushev, Nikolaj},
  journal={Proceedings on Privacy Enhancing Technologies},
  year={2021}
}

@inproceedings{de2014practical,
  title={Practical secure decision tree learning in a teletreatment application},
  author={De Hoogh, Sebastiaan and Schoenmakers, Berry and Chen, Ping and op den Akker, Harm},
  booktitle={Financial Cryptography and Data Security: 18th International Conference, FC 2014, Christ Church, Barbados, March 3-7, 2014, Revised Selected Papers 18},
  pages={179--194},
  year={2014},
  organization={Springer}
}

@article{wu2020privacy,
  title={Privacy preserving vertical federated learning for tree-based models},
  author={Wu, Yuncheng and Cai, Shaofeng and Xiao, Xiaokui and Chen, Gang and Ooi, Beng Chin},
  journal={Proceedings of the VLDB Endowment},
  volume={13},
  number={12},
  pages={2090--2103},
  year={2020},
  publisher={VLDB Endowment}
}

@inproceedings{beaver1995precomputing,
  title={Precomputing oblivious transfer},
  author={Beaver, Donald},
  booktitle={Annual International Cryptology Conference},
  pages={97--109},
  year={1995},
  organization={Springer}
}

@inproceedings{boyle2016function,
  title={Function secret sharing: Improvements and extensions},
  author={Boyle, Elette and Gilboa, Niv and Ishai, Yuval},
  booktitle={Proceedings of the 2016 ACM SIGSAC Conference on Computer and Communications Security},
  pages={1292--1303},
  year={2016}
}

@article{gupta2023sigma,
  title={Sigma: Secure gpt inference with function secret sharing},
  author={Gupta, Kanav and Jawalkar, Neha and Mukherjee, Ananta and Chandran, Nishanth and Gupta, Divya and Panwar, Ashish and Sharma, Rahul},
  journal={Cryptology ePrint Archive},
  year={2023}
}

@inproceedings{jawalkar2024orca,
  title={Orca: Fss-based secure training and inference with gpus},
  author={Jawalkar, Neha and Gupta, Kanav and Basu, Arkaprava and Chandran, Nishanth and Gupta, Divya and Sharma, Rahul},
  booktitle={2024 IEEE Symposium on Security and Privacy (SP)},
  pages={597--616},
  year={2024},
  organization={IEEE}
}

@article{canetti2000security,
  title={Security and composition of multiparty cryptographic protocols},
  author={Canetti, Ran},
  journal={Journal of CRYPTOLOGY},
  volume={13},
  pages={143--202},
  year={2000},
  publisher={Springer}
}

@article{lindell2017simulate,
  title={How to simulate it--a tutorial on the simulation proof technique},
  author={Lindell, Yehuda},
  journal={Tutorials on the Foundations of Cryptography: Dedicated to Oded Goldreich},
  pages={277--346},
  year={2017},
  publisher={Springer}
}

@incollection{goldreich2019play,
  title={How to play any mental game, or a completeness theorem for protocols with honest majority},
  author={Goldreich, Oded and Micali, Silvio and Wigderson, Avi},
  booktitle={Providing Sound Foundations for Cryptography: On the Work of Shafi Goldwasser and Silvio Micali},
  pages={307--328},
  year={2019}
}

@article{fu2019experimental,
  title={An experimental evaluation of large scale GBDT systems},
  author={Fu, Fangeheng and Jiang, Jiawei and Shao, Yingxia and Cui, Bin},
  journal={Proceedings of the VLDB Endowment},
  volume={12},
  number={11},
  pages={1357--1370},
  year={2019},
  publisher={VLDB Endowment}
}

@inproceedings{demmler2015aby,
  title={ABY-A framework for efficient mixed-protocol secure two-party computation.},
  author={Demmler, Daniel and Schneider, Thomas and Zohner, Michael},
  booktitle={NDSS},
  year={2015}
}

@article{paszke2019pytorch,
  title={Pytorch: An imperative style, high-performance deep learning library},
  author={Paszke, Adam and Gross, Sam and Massa, Francisco and Lerer, Adam and Bradbury, James and Chanan, Gregory and Killeen, Trevor and Lin, Zeming and Gimelshein, Natalia and Antiga, Luca and others},
  journal={Advances in neural information processing systems},
  volume={32},
  year={2019}
}

@inproceedings{bai2023mostree,
  title={Mostree: Malicious Secure Private Decision Tree Evaluation with Sublinear Communication},
  author={Bai, Jianli and Song, Xiangfu and Zhang, Xiaowu and Wang, Qifan and Cui, Shujie and Chang, Ee-Chien and Russello, Giovanni},
  booktitle={Proceedings of the 39th Annual Computer Security Applications Conference},
  pages={799--813},
  year={2023}
}

@article{wu2024ditto,
  title={Ditto: Quantization-aware Secure Inference of Transformers upon MPC},
  author={Wu, Haoqi and Fang, Wenjing and Zheng, Yancheng and Ma, Junming and Tan, Jin and Wang, Yinggui and Wang, Lei},
  journal={arXiv preprint arXiv:2405.05525},
  year={2024}
}

@inproceedings{liu2020boosting,
  title={Boosting privately: Federated extreme gradient boosting for mobile crowdsensing},
  author={Liu, Yang and Ma, Zhuo and Liu, Ximeng and Ma, Siqi and Nepal, Surya and Deng, Robert H and Ren, Kui},
  booktitle={2020 IEEE 40th international conference on distributed computing systems (ICDCS)},
  pages={1--11},
  year={2020},
  organization={IEEE}
}

@inproceedings{chen2016xgboost,
  title={Xgboost: A scalable tree boosting system},
  author={Chen, Tianqi and Guestrin, Carlos},
  booktitle={Proceedings of the 22nd acm sigkdd international conference on knowledge discovery and data mining},
  pages={785--794},
  year={2016}
}

@inproceedings{wang2022foster,
  title={Foster: Feature boosting and compression for class-incremental learning},
  author={Wang, Fu-Yun and Zhou, Da-Wei and Ye, Han-Jia and Zhan, De-Chuan},
  booktitle={European conference on computer vision},
  pages={398--414},
  year={2022},
  organization={Springer}
}

@article{ke2017lightgbm,
  title={Lightgbm: A highly efficient gradient boosting decision tree},
  author={Ke, Guolin and Meng, Qi and Finley, Thomas and Wang, Taifeng and Chen, Wei and Ma, Weidong and Ye, Qiwei and Liu, Tie-Yan},
  journal={Advances in neural information processing systems},
  volume={30},
  year={2017}
}

@article{wang2022corporate,
  title={Corporate finance risk prediction based on LightGBM},
  author={Wang, Di-ni and Li, Lang and Zhao, Da},
  journal={Information Sciences},
  volume={602},
  pages={259--268},
  year={2022},
  publisher={Elsevier}
}

@article{behera2022xgboost,
  title={XGBoost regression model-based electricity tariff plan recommendation in smart grid environment},
  author={Behera, Dayal Kumar and Das, Madhabananda and Swetanisha, Subhra and Nayak, Janmenjoy},
  journal={International Journal of Innovative Computing and Applications},
  volume={13},
  number={2},
  pages={79--87},
  year={2022},
  publisher={Inderscience Publishers (IEL)}
}

@article{feng2020xgboost,
  title={An XGBoost-based casualty prediction method for terrorist attacks},
  author={Feng, Yi and Wang, Dujuan and Yin, Yunqiang and Li, Zhiwu and Hu, Zhineng},
  journal={Complex \& Intelligent Systems},
  volume={6},
  pages={721--740},
  year={2020},
  publisher={Springer}
}

@article{shi2022quantized,
  title={Quantized training of gradient boosting decision trees},
  author={Shi, Yu and Ke, Guolin and Chen, Zhuoming and Zheng, Shuxin and Liu, Tie-Yan},
  journal={Advances in neural information processing systems},
  volume={35},
  pages={18822--18833},
  year={2022}
}

@article{el2023privatree,
  title={PrivaTree: Collaborative Privacy-Preserving Training of Decision Trees on Biomedical Data},
  author={El Zein, Yamane and Lemay, Mathieu and Huguenin, K{\'e}vin},
  journal={IEEE/ACM Transactions on Computational Biology and Bioinformatics},
  volume={21},
  number={1},
  pages={1--13},
  year={2023},
  publisher={IEEE}
}

@article{ercegovac2000improving,
  title={Improving Goldschmidt division, square root, and square root reciprocal},
  author={Ercegovac, Milos D and Imbert, Laurent and Matula, David W and Muller, J-M and Wei, Guoheng},
  journal={IEEE Transactions on Computers},
  volume={49},
  number={7},
  pages={759--763},
  year={2000},
  publisher={IEEE}
}

@article{zhao2022sgboost,
  title={SGBoost: An efficient and privacy-preserving vertical federated tree boosting framework},
  author={Zhao, Jiaqi and Zhu, Hui and Xu, Wei and Wang, Fengwei and Lu, Rongxing and Li, Hui},
  journal={IEEE Transactions on Information Forensics and Security},
  volume={18},
  pages={1022--1036},
  year={2022},
  publisher={IEEE}
}

@article{dai2024nodeguard,
  title={NodeGuard: A Highly Efficient Two-Party Computation Framework for Training Large-Scale Gradient Boosting Decision Tree},
  author={Dai, Tianxiang and Jiang, Yufan and Li, Yong and Mei, Fei},
  journal={Cryptology ePrint Archive},
  year={2024}
}
